%% file: TNMoesp.tex
\documentclass[twocolumn]{autart}    % Enable this line and disable the 
\usepackage[cmex10]{amsmath}
\usepackage{epsfig,amsfonts,amssymb,cite,times,subfigure,floatflt,wrapfig,enumerate,tikz,url,bm}
\usepackage{algorithmicx,algpseudocode}
\usepackage{booktabs}
\usepackage{siunitx}
\usepackage{graphicx}          % Include this line if your 
\newtheorem{example}{Example}
\newtheorem{theorem}{Theorem}[section]
\newtheorem{lemma}{Lemma}

\newtheorem{definition}{Definition}[section]
\newtheorem{proof}{Proof}

\newcommand{\mat}[1]{\bm{#1}}
\newcommand{\ten}[1]{\bm{\mathcal{#1}}}

\newcommand{\kpr}[1]{\textsuperscript{\textcircled{#1}}}

\begin{document}

\begin{frontmatter}

\title{Tensor network subspace identification of polynomial state space models}

%\thanks[footnoteinfo]{This paper was not presented at any IFAC 
%meeting. Corresponding author M.~T.~Cicero. Tel. +XXXIX-VI-mmmxxi. 
%Fax +XXXIX-VI-mmmxxv.}

\author[HKU]{Kim Batselier}\ead{kimb@eee.hku.hk},    % Add the 
\author[HKU]{Ching Yun Ko}\ead{cyko@eee.hku.hk},    % Add the 
\author[HKU]{Ngai Wong}\ead{nwong@eee.hku.hk},               % e-mail address 

\address[HKU]{The Department of Electrical and Electronic Engineering, The University of Hong Kong}  % Please supply                                              

                             % keyword list or with the 
                                          % help of the Automatica 
                                          % keyword wizard

\begin{abstract}                          % Abstract of not more than 200 words.
This article introduces a tensor network subspace algorithm for the identification of specific polynomial state space models. The polynomial nonlinearity in the state space model is completely written in terms of a tensor network, thus avoiding the curse of dimensionality. We also prove how the block Hankel data matrices in the subspace method can be exactly represented by low rank tensor networks, reducing the computational and storage complexity significantly. The performance and accuracy of our subspace identification algorithm are illustrated by numerical experiments, showing that our tensor network implementation is around 20 times faster than the standard matrix implementation before the latter fails due to insufficient memory, is robust with respect to noise and can model real-world systems.  
\end{abstract}

\end{frontmatter}

\section{Introduction}
Linear time-invariant (LTI) systems~\cite{Kailath} are a very useful framework for describing dynamical systems and have consequently been applied in myriad domains. Parametric system identification deals with the estimation of parameters for a given model structure from a set of measured vector input-output pairs $(\mat{u}_0,\mat{y}_0),\ldots,(\mat{u}_{L-1},\mat{y}_{L-1})$ and has been thoroughly studied in the 1980's and 1990's. Two important model structures for LTI systems are transfer function models and state space models, which can be converted into one another. The dominant framework for the estimation of transfer function models are prediction error and instrumental variables methods~\cite{ljung1999system,soderstrom1988system}, while state space models are typically estimated through subspace methods~\cite{katayama2005subspace,van2012subspace}. Prediction error methods are iterative methods that estimate the transfer function parameters such that the resulting prediction errors are minimized. These iterative methods suffer from some disadvantages such as no guaranteed convergence, sensitivity of the result on initial estimates and getting stuck in a local minimum of the objective function. Subspace methods, on the other hand, are non-iterative methods that rely on numerical linear algebra operations such as the singular value decomposition (SVD) or QR decomposition~\cite{matrixcomputations} of particular block Hankel data matrices. Estimates found through subspace methods are often good initial guesses for the iterative prediction error methods.

The most general nonlinear extension of the discrete-time linear state space model is
\begin{align*}
\mat{x}_{t+1} = f(\mat{x}_t,\mat{u}_t),\\
\mat{y}_{t} = g(\mat{x}_t,\mat{u}_t),
\end{align*}
where $\mat{x}_t,\mat{u}_t,\mat{y}_t$ are the state, input and output vectors at time $t$, respectively and $f(\cdot),g(\cdot)$ are nonlinear vector functions. By choosing different nonlinear functions $f(\cdot),g(\cdot)$ one effectively ends up with very different nonlinear state space models. A popular choice for the nonlinear functions are multivariate polynomials. The most general polynomial state space model, where $f(\cdot),g(\cdot)$ are multivariate polynomials in both the state and the input, is described in~\cite{PADUART2010647}. Being the most general form implies that it has a large expressive power, enabling the description of many different kinds of dynamics. This, however, comes at the cost of having to estimate an exponentially growing number of parameters as the degree of the polynomials increases. Furthermore, the identification method relies on solving a highly nonlinear optimization problem using iterative methods.

In~\cite{kruppa2014multilinear}, multilinear time invariant (MTI) systems are proposed. MTI systems are systems for which $f(\cdot),g(\cdot)$ are multivariate polynomial functions where each state or input variable is limited to a maximal degree of one. The number of parameters of an $m$-input-$p$-output MTI system with $n$ states is then $(n+p)2^{(n+m)}$, growing exponentially with the number of state and input variables. This curse of dimensionality is then effectively lifted with tensor methods. In this article, we propose the following polynomial state space model
\begin{align}
\label{eqn:model}
\nonumber \mat{x}_{t+1} = \mat{A}\,\mat{x}_t + f(\mat{u}_t),\\
\mat{y}_{t} = \mat{C}\,\mat{x}_t+ g(\mat{u}_t),
\end{align}
where both $f(\cdot),g(\cdot)$ are multivariate polynomials of total degree $d$. We then show that it is possible to identify these models using conventional subspace methods. The number of parameters that need to be estimated, however, will also grow exponentially with $d$. We then propose to use tensor networks to represent the polynomials $f(\cdot),g(\cdot)$ and modify the conventional MOESP subspace method to work for tensor networks. The main contributions of this article are:
\begin {enumerate}
\item We extend linear time-invariant state space models to a specific class of polynomial state space models with a linear state sequence and a polynomial input relation.
\item We modify the MOESP method~\cite{moesp1,moesp2} to utilize tensor networks for the identification of the proposed polynomial state space model. 
\item We prove in Theorem~\ref{theo:rankUTN} that the block Hankel data matrices in subspace methods are exactly represented by low-rank tensor networks, thereby reducing the computational and storage complexity significantly.
\end{enumerate}

% The main advantages of our model are:
% \begin{enumerate}
% \item In contrast to the MIMO Volterra systems described in~\cite{MVMALS,TNKalman}, the proposed model is a nonlinear extension of IIR filters.
% \item Identical to LTI systems, the stability of our polynomial model is easily determined from the eigenvalues of the $\mat{A}$ matrix.
% \item The polynomial model can be identified through subspace identification methods, which are non-iterative methods that do not suffer from the disadvantages of iterative optimization routines (e.g. no guaranteed convergence, local minima of the objective criterion and sensitivity to initial estimates).
% \end{enumerate}
The main outline of this article is as follows. First, we briefly discuss some tensor network preliminaries in Section~\ref{sec:prelim}. The proposed polynomial state space model is discussed in detail in Section~\ref{sec:pss}. The development and implementation of our proposed tensor network subspace identification method is described in Section~\ref{sec:tnmoesp}. The algorithm to simulate our proposed polynomial state space model in tensor network form is given in Section~\ref{sec:sim}. Numerical experiments validate and demonstrate the efficacy of our tensor network subspace identification method in Section~\ref{sec:experiments}. All our algorithms were implemented in the MATLAB/Octave TNMOESP package and can be freely downloaded from \url{https://github.com/kbatseli/TNMOESP}. Finally, some conclusions and future work are formulated in Section~\ref{sec:conclusions}.

\section{Preliminaries}
\label{sec:prelim}
Most of the notation on subspace methods is adopted from~\cite{katayama2005subspace} and the notation on tensors from~~\cite{MVMALS,TNKalman} is also used. Tensors are multi-dimensional arrays that generalize the notions of vectors and matrices to higher orders. A $d$-way or $d$th-order tensor is denoted $\ten{A} \in \mathbb{R}^{n_1 \times  n_2 \times \cdots \times n_d}$ and hence each of its entries is determined by $d$ indices $i_1,\ldots,i_d$. We use the MATLAB convention that indices  start from 1, such that $1 \leq i_k \leq n_k\, (k=1,\ldots,d)$. The numbers $n_1,n_2,\ldots,n_d$ are called the dimensions of the tensor. For practical purposes, only real tensors are considered. We use boldface capital calligraphic letters $\ten{A},\ten{B},\ldots$ to denote tensors, boldface capital letters $\mat{A},\mat{B},\ldots$ to denote matrices, boldface letters $\mat{a},\mat{b},\ldots$ to denote vectors, and Roman letters $a,b,\ldots$ to denote scalars. The elements of a set of $d$ tensors, in particular in the context of tensor networks, are denoted $\ten{A}^{(1)},\ten{A}^{(2)},\ldots,\ten{A}^{(d)}$. The transpose of a matrix $\mat{A}$ or vector $\mat{a}$ are denoted $\mat{A}^T$ and $\mat{a}^T$, respectively. The unit matrix of order $n$ is denoted $\mat{I}_n$. A matrix with all zero entries is denoted $\mat{O}$.

A very useful graphical representation of scalars, vectors, matrices and tensors is shown in Figure~\ref{fig:TN}. The number of unconnected edges of each node represents the order of the corresponding tensor. Scalars therefore are represented by nodes without any edges, while a matrix is represented by a node that has two edges. This graphical representation allows us to visualize the different tensor networks and operations in this article in a very straightforward way. We also adopt the MATLAB notation regarding entries of tensors, e.g. $\mat{A}(:,1)$ denotes the first column of the matrix $\mat{A}$. 
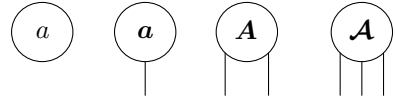
\begin{figure}[tb]
\begin{center}
\input{figures/TNgraphs.tex}
\caption{Graphical depiction of a scalar $a$, vector $\mat{a}$, matrix $\mat{A}$ and 3-way tensor $\ten{A}$.}
\label{fig:TN}
\end{center}
\end{figure}
We now give a brief description of some required tensor operations. The generalization of the matrix-matrix multiplication to tensors involves a multiplication of a matrix with a $d$-way tensor along one of its $d$ possible modes.
\begin{definition}(\cite[p.~460]{tensorreview})
The $k$-mode product of a tensor $\ten{A}\in\mathbb{R}^{n_1\times\cdots \times n_k\times\cdots\times n_d}$ with a matrix $\bm{U}\in\mathbb{R}^{p_k\times n_k}$ is denoted \mbox{$\ten{B}=\ten{A}\, {\times_k}\, \bm{U}$} and defined by%
\begin{align}
\nonumber \ten{B}(i_1,\cdots,i_{k-1},j,i_{k+1}, \cdots, i_d) &= \hfill \\
\sum\limits_{i_k=1}^{n_k}  \mat{U}(j,i_k) \mat{A}(i_1,\cdots,i_{k-1},i_k,&i_{k+1},\cdots,i_d),%
\label{eqn:kmode}
\end{align}%
with $\ten{B}\in\mathbb{R}^{n_1\times\cdots \times n_{k-1}\times p_k\times n_{k+1}\times\cdots\times n_d}$.
\end{definition}
For a $(d+1)$-way tensor $\ten{A}  \in \mathbb{R}^{n \times  m \times \cdots \times m}$ and vector $\mat{x} \in \mathbb{R}^{m}$, we define the short hand notation for the vector
\begin{align*}
\ten{A}\,\mat{x}^d &:= \ten{A} \times_2 \mat{x}^T \times_3 \cdots \times_{d+1} \mat{x}^T \; \in \mathbb{R}^n.
\end{align*}
The Kronecker product will be repeatedly used to describe the polynomial nonlinearity.
\begin{definition}{\textbf{(Kronecker product)}}
If $\mat{B} \in \mathbb{R}^{m_1 \times m_2}$ and~$\mat{C} \in \mathbb{R}^{n_1 \times n_2}$, then their Kronecker product $\mat{B} \otimes \mat{C}$ is the $m_1n_1 \times m_2n_2$ matrix
\begin{equation}
\begin{pmatrix}
b_{11} & \cdots & b_{1m_2}\\
\vdots & \ddots & \vdots \\
b_{m_11} & \cdots &b_{m_1m_2}\\
\end{pmatrix} \otimes \mat{C} \;=\;
\begin{pmatrix}
b_{11}\mat{C} & \cdots & b_{1m_2}\mat{C}\\
\vdots & \ddots & \vdots \\
b_{m_11}\mat{C} & \cdots & b_{m_1m_2}\mat{C}\\
\end{pmatrix}.
\label{def:kron}
\end{equation}
\end{definition}
% The following definitions of the Kronecker product of two tensors generalize the matrix Kronecker product.
% \begin{definition}
% Suppose we have two $d$-way tensors \mbox{$\ten{A} \in \mathbb{R}^{n_1 \times \cdots \times n_d}$}, $\mat{B} \in \mathbb{R}^{m_1 \times \cdots \times m_d}$. The right Kronecker product $\ten{C}=\ten{A} \otimes \ten{B} \in \mathbb{R}^{n_1m_1 \times \cdots \times n_dm_d}$ is then a $d$-way tensor such that
% \begin{align*}
% \ten{C}([j_1i_1],\ldots,[j_di_d]) &= \ten{A}(i_1,\ldots,i_d)\, \ten{B}(j_1,\ldots,j_d),
% \end{align*}
% with $[j_ki_k] = j_k + (i_k-1)\,m_k$ for $k=1,\ldots,d$. The left Kronecker product $\ten{A} \otimes_L \ten{B}$ is a $d$-way tensor $\ten{C}$ such that
% \begin{align*}
% \ten{C}([i_1j_1],\ldots,[i_dj_d]) &= \ten{A}(i_1,\ldots,i_d)\, \ten{B}(j_1,\ldots,j_d),
% \end{align*}
% with $[i_kj_k] = i_k + (j_k-1)\,n_k$ for $k=1,\ldots,d$. 
% \end{definition}
% Note that we have that $\ten{A} \otimes \ten{B} = \ten{B} \otimes_L \ten{A}$.
\begin{definition}
The Khatri-Rao product $\mat{A} \odot \mat{B}$ between $\mat{A} \in \mathbb{R}^{n_1 \times p}$ and $\mat{B} \in \mathbb{R}^{n_2 \times p}$ is the matrix $\mat{C} \in \mathbb{R}^{n_1n_2 \times p}$ with
\begin{align*}
\mat{C}(:,k) &= \mat{A}(:,k) \otimes \mat{B}(:,k), \, (k=1,\ldots,p).
% \mat{C}(:,k) &= \mat{A}(:,k) \otimes_L \mat{B}(:,k), \, (k=1,\ldots,p).
\end{align*}
\end{definition}
Another common operation on tensors that we will use throughout this article is reshaping.
\begin{definition} We adopt the MATLAB/Octave reshape operator ``reshape($\ten{A},[n_1,n_2,n_3 \cdots])$", which reshapes the $d$-way tensor $\ten{A}$ into a tensor with dimensions $n_1 \times n_2 \times \cdots \times n_d$. The total number of elements of $\ten{A}$ must be the same as~$n_1\times n_2 \times \cdots \times n_d$.
\end{definition}
% \begin{example}
% \label{ex:ex1}
% We illustrate the reshaping operator on~the~$4\times 3 \times 2$ tensor $\ten{A}$ that contains all entries from 1 up to 24.
% \begin{align*}
% \textrm{reshape}(\ten{A},[4,6]) &= 
% \begin{pmatrix}
% 1 & 5 & 9 & 13 & 17 & 21\\
% 2 & 6 & 10 & 14 &18 & 22\\
% 3 & 7 & 11 &15 &19 & 23\\
% 4 & 8 & 12 & 16&20 & 24
% \end{pmatrix}.
% \end{align*}
% \end{example}
Storing all entries of a $d$-way tensor with dimension size $n$ requires $n^d$ storage units and quickly becomes prohibitively large for increasing values of $n$ and $d$. When the data in the tensor have redundancies, then more economic ways of storing the tensor exist in the form of tensor decompositions. The tensor decomposition used throughout this article is a particular tensor network called matrix product operators, also called tensor train matrices in the applied math community~\cite{ivanTT,TNorus,Oseledets2010}. Suppose we have an $n \times m$ matrix $\mat{A}$, where the row index $i$ can be written as a grouped index $[i_1i_2\cdots i_d]$ such that
\begin{align*}
i &= i_1 + (i_2-1)\,n_1 + \cdots + (i_d-1)\,\prod_{j=1}^{d-1}n_j,
\end{align*}
and likewise for the column index $j=[j_1j_2\cdots j_d]$. This implies that
\begin{align*}
1 &\leq i_k \leq n_k, \;  1 \leq j_k \leq m_k \,(k=1,\ldots,d),
\end{align*}
% \begin{align*}
% 1 &\leq i_1 \leq n_1, &  1 \leq j_1 \leq m_1,\\
% 1 &\leq i_2 \leq n_2, &  1 \leq j_2 \leq m_2,\\
%   &    \vdots   & \vdots \\
%  1 &\leq i_d \leq n_d, &   1 \leq j_d \leq m_d,
% \end{align*}
and $n=n_1n_2\cdots n_d$, $m=m_1m_2\cdots m_d$. A matrix product operator is then a representation of the matrix $\mat{A}$ as a set of 4-way tensors $\ten{A}^{(k)} \in \mathbb{R}^{r_k \times i_k \times j_k \times r_{k+1}} (k=1,\ldots,d)$ such that each matrix entry $\mat{A}([i_1i_2\cdots i_d],[j_1j_2\cdots j_d])$ is per definition
{\scriptsize
\begin{align}
\label{eqn:mpodef}
\sum_{k_2,\ldots,k_d}&\ten{A}^{(1)}(1,i_1,j_2,k_2)\ten{A}^{(2)}(k_2,i_2,j_2,k_3)\cdots \ten{A}^{(d)}(k_d,i_d,j_d,1).
\end{align}}
Note that $r_1=r_{d+1}=1$ is required in order for the summation~\eqref{eqn:mpodef} to result in a scalar. A graphical representation of the matrix product structure is shown in Figure~\ref{fig:mpo} for $d=4$. The fully connected edges represent the auxiliary indices $k_2,\ldots,k_d$ that are summed over. Figure~\ref{fig:mpo} illustrates the power of this particular visual representation by replacing the complicated summation in~\eqref{eqn:mpodef} with a simple graph. The canonical tensor network ranks $r_1,r_2,\ldots,r_{d+1}$ are defined as the minimal values such that the summation in~\eqref{eqn:mpodef} is exactly equal to $\mat{A}([i_1i_2\cdots i_d],[j_1j_2\cdots j_d])$. An upper bound on the canonical rank $r_k$ for a matrix product operator of $d$ tensors for which $r_1=r_{d+1}=1$ is given by the following theorem.
\begin{theorem}(Modified version of Theorem 2.1 in \cite{ttcross})
\label{theo:MPTranks}
For any matrix $\mat{A} \in \mathbb{R}^{n_1n_2\cdots n_d \times m_1m_2\cdots m_d }$ there exists a matrix product operator with ranks $r_1=r_{d+1}=1$ such that the canonical ranks $r_k$ satisfy
\begin{align*}
r_k \leq \textrm{min}\,\left( \prod_{i=1}^{k-1} n_im_i, \prod_{i=k}^d n_im_i \right) \textrm{ for } k=2,\ldots,d.
\end{align*}
\end{theorem}

Note that using a matrix product operator structure can reduce the storage cost of a square $n^d \times n^d$ matrix from $n^{2d}$ down to approximately $dn^2r^2$, where $r$ is the maximal tensor network rank.
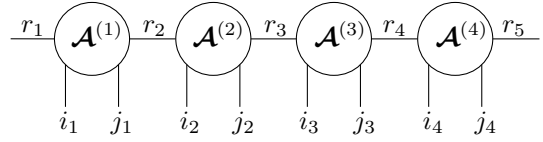
\begin{figure}[tb]
\begin{center}
\input{figures/MPO.tex}
\caption{Graphical depiction of a tensor network that consists of four 4-way tensors $\ten{A}^{(1)},\ldots,\ten{A}^{(4)}$.}
\label{fig:mpo}
\end{center}
\end{figure}

\section{Polynomial state space model}
\label{sec:pss}
\subsection{The model}
We rewrite our proposed polynomial state space model~\eqref{eqn:model} in terms of tensors as
\begin{align}
\label{eqn:pss}
\nonumber \mat{x}_{t+1} &= \mat{A}\,\mat{x}_{t} + \ten{B}\, \mat{u}_t^d,\\
\mat{y}_t &= \mat{C}\,\mat{x}_{t} + \ten{D}\, \mat{u}_t^d,
\end{align}
where $\mat{x}_t \in \mathbb{R}^{n}$ is the state vector and $\mat{y}_t \in \mathbb{R}^{p}, \mat{u}_t \in \mathbb{R}^{m}$ are the output and input vectors, respectively. The matrices $\mat{A} \in \mathbb{R}^{n\times n}$ and $\mat{C} \in \mathbb{R}^{p \times n}$ are the regular state transition and output model matrices from LTI systems. The main difference between the model~\eqref{eqn:pss} and LTI systems are the $\ten{B}\, \mat{u}_t^d$ and $\ten{D}\, \mat{u}_t^d$ terms, where $\ten{B}$ is a $(d+1)$- way tensor with dimensions $n \times m \times m \times \cdots \times m$ and $\ten{D}$ is a $(d+1)$-way tensor with dimensions $p \times m \times m \times \cdots \times m$. The input vector $\mat{u}_t$ is defined as
\begin{align*}
\mat{u}_t := \begin{pmatrix}1 & u_t^{(1)} & u_t^{(2)} & \cdots & u_t^{(m-1)}\end{pmatrix}^T,
\end{align*}
which implies that there are $m-1$ measured input signals. An alternative way to write~\eqref{eqn:pss} is
\begin{align}
\label{eqn:pssmatricized}
\nonumber \mat{x}_{t+1} &= \mat{A}\,\mat{x}_{t} + \mat{B}\, \mat{u}_t\kpr{d},\\
\mat{y}_t &= \mat{C}\,\mat{x}_{t} + \mat{D}\, \mat{u}_t\kpr{d},
\end{align}
where $\mat{B},\mat{D}$ are the tensors $\ten{B},\ten{D}$ reshaped into $n \times m^d$ and $p \times m^d$ matrices, respectively, and $\mat{u}_t\kpr{d}$ is defined as the $d$-times repeated left Kronecker product
\begin{align}
\mat{u}_t\kpr{d} := \overbrace { {\mat{u}_t\otimes \mat{u}_t\otimes \cdots \otimes \mat{u}_t}}^{d}\; \in\; \mathbb{R}^{m^d}.
% \mat{u}_t\kpr{d} := \overbrace { {\mat{u}_t\otimes_L \mat{u}_t\otimes_L \cdots \otimes_L \mat{u}_t}}^{d}\; \in\; \mathbb{R}^{m^d}.
\end{align}
Each row of $\mat{B},\mat{D}$ can therefore be interpreted as containing the coefficients of a multivariate polynomial of total degree $d$ in the $m-1$ inputs $u_t^{(1)}, u_t^{(2)}, \ldots , u_t^{(m-1)}$. Note that the affine terms, also called the constant terms, of both $\mat{B}\mat{u}_{i}\kpr{d}$ and $\mat{D}\mat{u}_{i}\kpr{d}$ are defined to be exactly zero. 
\subsection{Internal stability and persistence of excitation}
Repeated application of~\eqref{eqn:pssmatricized} for $t=0,\ldots,t-1$ allows us to write
\begin{align}
\label{eqn:singleoutput}
\mat{y}_t &= \mat{C}\mat{A}^t\,\mat{x}_0 + \sum_{i=0}^{t-1}\,\mat{C}\mat{A}^{t-1-i}\mat{B}\mat{u}_{i}\kpr{d}+\mat{D}\mat{u}_{t}\kpr{d}.
\end{align}
The current output $\mat{y}_t$ at time $t$ is therefore completely determined by the initial state $\mat{x}_0$ and all input signals $\mat{u}_0,\ldots,\mat{u}_t$. Since the state sequence is linear when a zero input is applied, the condition for internal stability of~\eqref{eqn:pss} is identical to LTI systems.
\begin{lemma}
The polynomial state space model~\eqref{eqn:pss} is internally stable when all eigenvalues of $\mat{A}$ satisfy $|\lambda_i| < 1$, $i=1,\ldots,n$.
\end{lemma}
A stable polynomial state space model then implies that the transient part $\mat{C}\mat{A}^t\,\mat{x}_0$ will have a progressively smaller contribution to $\mat{y}_t$ as $t$ increases. Each term of the sum $\sum_{i=0}^{t-1}\,\mat{C}\mat{A}^{t-1-i}\mat{B}\mat{u}_{i}\kpr{d}$ can be interpreted as an $m$-variate polynomial of total degree $d$ in $\mat{u}_i$. Writing out~\eqref{eqn:singleoutput} for $t=0,\ldots,k-1$ we obtain
\begin{align}
\label{eqn:outputcol}
\mat{y}_{0|k-1} &=    
\mat{O}_k \; \mat{x}_0 +  \mat{P}_k \; \mat{u}_{0|k-1},
\end{align}
where
\begin{align*}
\mat{y}_{0|k-1} := \begin{pmatrix} \mat{y}_0 \\ \mat{y}_1 \\ \mat{y}_2 \\ \vdots \\ \mat{y}_{k-1} \end{pmatrix} \in \mathbb{R}^{kp}, \; \mat{u}_{0|k-1}:=\begin{pmatrix} \mat{u}_0\kpr{d} \\ \mat{u}_1\kpr{d}\\ \mat{u}_2\kpr{d} \\ \vdots \\ \mat{u}_{k-1}\kpr{d} \end{pmatrix} \in \mathbb{R}^{km^d},
\end{align*}
and
\begin{align}
\mat{O}_{k} := \begin{pmatrix} \mat{C} \\ \mat{CA} \\ \mat{C}\mat{A}^2   \\ \vdots \\ \mat{C}\mat{A}^{k-1}   \end{pmatrix}  \in \mathbb{R}^{kp \times n}
\end{align}
is the well-known extended observability matrix of linear time-invariant systems and
\begin{align}
\mat{P}_k := \begin{pmatrix} \mat{D} &    &   &   &   \\ 
\mat{CB} & \mat{D} &   &   \\
\vdots  &  \ddots & \ddots &  \\
\mat{C}\mat{A}^{k-2}\mat{B} & \cdots  & \mat{CB}  & \mat{D}  \\
\end{pmatrix} \in  \mathbb{R}^{kp \times km^d}
\end{align}
is a block Toeplitz matrix with an exponential number of columns. Following the idea described in~\cite{moonensubspace}, we write~\eqref{eqn:outputcol} in terms of block Hankel data matrices
\begin{align*}
\mat{Y}_{0|k-1} &:= 
\begin{pmatrix}
\mat{y}_0  & \mat{y}_1  & \cdots &\mat{y}_{N-1}\\
\mat{y}_1  & \mat{y}_2  & \cdots & \mat{y}_{N}\\
\vdots & \vdots &  & \vdots \\
\mat{y}_{k-1}  & \mat{y}_k  & \cdots & \mat{y}_{N+k-2}
\end{pmatrix} \in \mathbb{R}^{kp \times N},\\
% \end{align*}
% \begin{align*}
\mat{U}_{0|k-1}&:= 
\begin{pmatrix}
\mat{u}_0\kpr{d}  & \mat{u}_1\kpr{d}  & \cdots &\mat{u}_{N-1}\kpr{d}\\
\mat{u}_1\kpr{d}  & \mat{u}_2\kpr{d}  & \cdots &\mat{u}_{N}\kpr{d}\\
\vdots & \vdots &  & \vdots \\
\mat{u}_{k-1}\kpr{d}  & \mat{u}_k\kpr{d}  & \cdots &\mat{u}_{N+k-2}\kpr{d}
\end{pmatrix} \in \mathbb{R}^{km^d \times N}
\end{align*}
as
\begin{align}
\label{eqn:blockhankel}
\mat{Y}_{0|k-1} &=    
\mat{O}_k \; \mat{X} +  \mat{P}_k \; \mat{U}_{0|k-1},
\end{align}
where $\mat{X}:=\begin{pmatrix}\mat{x}_0& \mat{x}_1 & \cdots & \mat{x}_{N-1}\end{pmatrix} \in \mathbb{R}^{n\times N}$ is the state sequence matrix and $N$ is sufficiently large. Equation~\eqref{eqn:blockhankel} lies at the heart of subspace identification methods. We now introduce the notion of\textit{ persistence of excitation} for our polynomial state space model.
\begin{theorem}
\label{theo:rankU}
The rank of the block Hankel $km^d \times N$ matrix $\mat{U}_{0|k-1}$ is upper bounded by $r:=k{d+m-1 \choose m-1}-k+1$.
\end{theorem}
\begin{proof}
Consider the first block of $m^d$ rows of $\mat{U}_{0|k-1}$. Due to its repeated Kronecker product structure, this block has ${d+m-1 \choose m-1}$ distinct rows, which serves as an upper bound for its rank. There are $k$ such row blocks in $\mat{U}_{0|k-1}$, implying the upper bound for the rank is $k{d+m-1 \choose m-1}$. Each of the $k$ blocks, however, has one row that consists entirely of 1s. Only one such row contributes to the rank and the upper bound for the rank of $\mat{U}_{0|k-1}$ is therefore $k{d+m-1 \choose m-1}-k+1$.
\end{proof}
\begin{definition}
The input of a polynomial state space system~\eqref{eqn:pss} of total degree $d$ is persistent exciting of order $k$ if and only if $\textrm{rank}(\mat{U}_{0|k})=k{d+m-1 \choose m-1}-k+1$.
\end{definition}

\section{Tensor network subspace identification}
\label{sec:tnmoesp}
As the input-output relationship of our polynomial state space model satisfies~\eqref{eqn:blockhankel}, any subspace method can be applied in principle for the identification of the $\mat{A},\mat{B},\mat{C},\mat{D}$ matrices. Two candidates are the N4SID algorithm by Van Overschee and De Moor~\cite{n4sid,van2012subspace} and the MOESP algorithm by Verhaegen and Dewilde~\cite{moesp1,moesp2}.

Of particular concern are the $\mat{B}$ and $\mat{D}$ matrices, which have an exponential number of coefficients that need to be estimated. For moderate values of $m$ and $d$ one could still use a conventional matrix based implementation. It is possible, however, that it becomes impractical to store $\mat{B},\mat{D}$ in memory for large values of both $m$ and $d$. Our solution to this problem is to compute and store a tensor decomposition of $\ten{B},\ten{D}$ instead. More specifically, all entries of $\ten{B},\ten{D}$ can be computed from a tensor network. This implies that the subspace algorithm needs to be modified such that all computations can be performed on tensor networks. The MOESP algorithm in particular lends itself very well to such a modification. The main problem with N4SID is that it first estimates a state sequence and then recovers the $\mat{A},\mat{B},\mat{C},\mat{D}$ matrices in a single step by solving a linear system. This last step becomes problematic when $\mat{B}$ and $\mat{D}$ are represented by a tensor decomposition. The MOESP method, on the other hand, estimates both $\mat{A},\mat{C}$ and $\mat{B},\mat{D}$ in separate steps. The conventional MOESP algorithm is fully described in Algorithm~\ref{alg:moesp}. Before going into the required tensor network modifications of Algorithm~\ref{alg:moesp} in detail, we first discuss a few assumptions. 
\begin{alg}Conventional MOESP algorithm~\cite[p.~159]{katayama2005subspace}\\
\label{alg:moesp}
\textit{\textbf{Input}}: $L$ samples $(\mat{u}_0,\mat{y}_0),\ldots,(\mat{u}_{L-1},\mat{y}_{L-1})$, $k$\\
\textit{\textbf{Output}}:\makebox[0pt][l]{ Matrices $\mat{A},\mat{B},\mat{C},\mat{D}$}
\begin{algorithmic}[1]
\State $\begin{pmatrix}\mat{U}_{0|k-1} \\ \mat{Y}_{0|k-1}\end{pmatrix} = \begin{pmatrix}L_{11} & \mat{0} \\ L_{21} & L_{22} \end{pmatrix} \, \begin{pmatrix}\mat{Q}_1^T\\ \mat{Q}_2^T \end{pmatrix}$
\State SVD of $L_{22} = \begin{pmatrix}\mat{U}_1 & \mat{U}_2 \end{pmatrix} \begin{pmatrix} \mat{S}_1 & \mat{0}\\\mat{0} & \mat{0} \end{pmatrix} \begin{pmatrix} \mat{V}_1^T \\ \mat{V}_2^T\end{pmatrix}$
\State Define system order as $n:=\textrm{rank}(L_{22})$
\State $\mat{O}_k=\mat{U}_1\,\mat{S}_1^{1/2}$ and $\mat{C} = \mat{O}_k(1:p,:)$
\State Compute $\mat{A}$ from $\mat{O}_k(1:kp-p,:)\mat{A}=\mat{O}(p+1:kp,:)$
\State Partition $\mat{U}_2^T:=\begin{pmatrix}\mat{L}_1 & \cdots & \mat{L}_k \end{pmatrix}$ into $k$ blocks of size \mbox{$(kp-n) \times p$}.
\State Partition $\mat{U}_2^T\mat{L}_{21}\mat{L}_{11}^{-1}:=\begin{pmatrix}\mat{M}_1 & \cdots & \mat{M}_k \end{pmatrix}$ into $k$ blocks of size \mbox{$(kp-n)\times m^d$}
\State Define $\bar{\mat{L}}_i:= \begin{pmatrix}\mat{L}_i & \cdots & \mat{L}_k \end{pmatrix}$, $i=2,\ldots k$
\State Compute $\mat{B},\mat{D}$ from
\begin{align}
\label{eqn:BDsystem}
\begin{pmatrix}
\mat{L}_1& \bar{\mat{L}}_2\mat{O}_{k-1}\\
\mat{L}_2& \bar{\mat{L}}_3\mat{O}_{k-2}\\
\vdots & \vdots \\
\mat{L}_{k-1}& \bar{\mat{L}}_k\mat{O}_{1}\\
\mat{L}_k & \mat{0}
\end{pmatrix} \begin{pmatrix} \mat{D} \\ \mat{B} \end{pmatrix} &= \begin{pmatrix}\mat{M}_1 \\ \mat{M}_2 \\ \vdots \\ \mat{M}_{k-1} \\ \mat{M}_k \end{pmatrix}
\end{align}
\end{algorithmic}
\end{alg}
\subsection{Assumptions}
The classical assumptions for the applicability of subspace methods~\cite[p.~151]{katayama2005subspace} are repeated here with a small modification related to the persistence of excitation of the inputs. We also need two additional assumptions specifically related to tensor networks.
\begin{itemize}
\item Assumption 1: $\textrm{rank}(\mat{X}) = n$.
\item Assumption 2: $\textrm{rank}(\mat{U}_{0|k-1})=r=k{d+m-1 \choose m-1}-k+1$.
\item Assumption 3: $\textrm{row}(\mat{X})\; \cap\; \textrm{row}(\mat{U}_{0|k-1}) = \{0\}$, where $\textrm{row}(\cdot)$ denotes the row space of a matrix.
\item Assumption 4: $N,n,p \ll m^d$.
\item Assumption 5: $N = r + kp,\, L = N + k-1$.
\end{itemize}
Assumption 1 implies that the polynomial state space system is reachable\footnote{The reachability problem of our polynomial state space model can also be solved using methods based on Algorithm~\ref{alg:TNSVD}.}, which means that an initial zero state vector $\mat{x}_0=0$ can be transferred to any state in $\mathbb{R}^n$ by means of a sequence of control vectors $\mat{u}_0\kpr{d},\ldots,\mat{u}_{n-1}\kpr{d}$. Assumption 2 is the modified persistence of excitation condition of the input. Assumption 3 implies that the input-output data are obtained from an open-loop experiment, which implies without having any feedback system. Assumptions 4 and 5 imply that the block Hankel matrix $\mat{Y}_{0|k-1}$ can be explicitly constructed and does not require to be stored as a tensor network. Note that $r$ is the rank of $\mat{U}_{0|k-1}$ as given by Theorem~\ref{theo:rankU}. Assumption 5 also ensures that the $L_{11}$ factor that we will compute is of full rank and also allows us to compute $k$ for a given set of $L$ measurements. Indeed, since $L=N+k-1$, this implies that $k=L/(p+{d+m-1 \choose m-1})$.

\subsection{Construction of $\mat{U}_{0|k-1}$ tensor network}
The matrix $\mat{U}_{0|k-1}$ has dimensions $km^d \times N$ and therefore needs to be stored as a tensor network. Fortunately, its block Hankel structure will result in very small tensor network ranks. The most common methods to construct a tensor network are either the TT-SVD algorithm~\cite[p.~2301]{ivanTT} or TT-cross algorithm~\cite[p.~82]{ttcross}. These algorithms are however computationally too expensive as they neither take the block Hankel nor the repeated Khatri-Rao product structure of $\mat{U}_{0|k-1}$ into account. The main idea to convert $\mat{U}_{0|k-1}$ into its tensor network is to realize that the matrix
\begin{align*}
\mat{U}&:=\textrm{reshape}(\mat{U}_{0|k-1},[m^d,kN]),
\end{align*}
consists of $kN$ columns, each of which is a repeated left Kronecker product. In other words, if we define the matrix
\begin{align*}
\tilde{\mat{U}} &:= \begin{pmatrix}\mat{u}_0 & \mat{u}_1 & \cdots & \mat{u}_{k-1} & \mat{u}_1 & \mat{u}_2 & \cdots \mat{u}_{N+k-2}\end{pmatrix} \in \mathbb{R}^{m \times kN},
\end{align*}
then
\begin{align}
\label{eqn:Uodot}
\mat{U} &= \overbrace{\tilde{\mat{U}} \odot \tilde{\mat{U}} \odot \cdots \odot \tilde{\mat{U}}   }^d.
\end{align}
% Our algorithm to construct the tensor network of $\mat{U}_{0|k-1}$ directly relies on the following lemma.
% \begin{lemma}
% \label{lemma:Usum}
% Let $\mat{e}_i$ be the $i$th canonical basis vector in $\mathbb{R}^k$ and partition
% \begin{align*}
% \mat{U}_{0|k-1} &=
% \begin{pmatrix}
% \mat{U}_0\\
% \mat{U}_1\\
% \vdots \\
% \mat{U}_{k-1}
% \end{pmatrix} \textrm{with} \; \mat{U}_i \in \mathbb{R}^{m^d \times N} (i=0,\ldots,k-1).
% \end{align*}
% Then
% \begin{align}
% \label{eqn:Usum}
% \mat{U}_{0|k-1} &= \mat{e}_1 \otimes \mat{U}_0 + \mat{e}_2 \otimes \mat{U}_1 + \cdots + \mat{e}_k \otimes \mat{U}_{k-1}.
% \end{align}
% \end{lemma}
% The main idea is then to construct the tensor network for each of the $\mat{U}_0,\cdots,\mat{U}_{k-1}$ matrices and then compute~\eqref{eqn:Usum} directly with tensor networks. We now present an efficient algorithm to construct the tensor networks for each of the $\mat{U}_0,\cdots,\mat{U}_{k-1}$ matrices. Suppose we want to convert $\mat{U}_i\,(i=0,\ldots,k-1)$ into a tensor network. If we define
% \begin{align*}
% \tilde{\mat{U}}_i &:= \begin{pmatrix} \mat{u}_i & \mat{u}_{i+1} & \cdots & \mat{u}_{N-1+i} \end{pmatrix} \in \mathbb{R}^{m \times N}
% \end{align*}
% then
% \begin{align}
% \label{eqn:Uodot}
% \mat{U}_i := \overbrace {\tilde{\mat{U}}_i \odot \tilde{\mat{U}}_i \odot \cdots \odot \tilde{\mat{U}}_i }^d.
% \end{align}
It is possible to construct the tensor network for $\mat{U}$ quite efficiently using $(d-1)$ SVDs and Khatri-Rao products, described in pseudocode as Algorithm~\ref{alg:matrix2TN}. The desired tensor network is constructed starting with the first tensor $\ten{U}^{(1)}$ and proceeds up to $\ten{U}^{(d)}$. The correctness of Algorithm~\ref{alg:matrix2TN} is confirmed as the algorithm consists of computing a Khatri-Rao product followed by an SVD to determine the tensor network rank. The most expensive computational step in Algorithm~\ref{alg:matrix2TN} is the SVD with a computational complexity of approximately $O(r_jm^3k^2N^2)$ flops~\cite[p.~254]{matrixcomputations}. A graphical representation of the obtained tensor network for $\mat{U}$ with all dimensions labeled is shown in Figure~\ref{fig:UiTN}. Algorithm~\ref{alg:matrix2TN} is implemented in the TNMOESP MATLAB package as rkh2tn.m.
\begin{alg}Convert repeated Khatri-Rao product matrix into tensor network.\\
\label{alg:matrix2TN}
\textit{\textbf{Input}}: $m \times kN$ matrix $\tilde{\mat{U}}$, factor $d$\\
\textit{\textbf{Output}}:\makebox[0pt][l]{ tensor network $\ten{U}^{(1)},\ldots,\ten{U}^{(d)}$ of \eqref{eqn:Uodot}}
\begin{algorithmic}[1]
\State $\ten{U}^{(1)} \gets \textrm{reshape}(\tilde{\mat{U}},[1,m,kN,1])$
\For{$j=1,\ldots,d-1$}
\State $\mat{T} \gets \textrm{reshape}(\ten{U}^{(j)},[r_jm,kN])$ \hfill \%\,$r_1=1$
\State $\mat{T} \gets \mat{T} \odot \tilde{\mat{U}} $
\State $\mat{T} \gets \textrm{reshape}(\mat{T},[r_jm,mkN])$
\State $[\mat{U},\mat{S},\mat{V}] \gets \textrm{SVD}(\mat{T})$
\State $r_{j+1} \gets $ numerical rank of $\mat{T}$ determined from SVD
\State $\ten{U}^{(j)} \gets \textrm{reshape}(\mat{U},[r_j,m,1,r_{j+1}])$
\State $\ten{U}^{(j+1)} \gets \textrm{reshape}(\mat{S}\mat{V}^T,[r_{j+1},m,kN,1])$
\EndFor
\end{algorithmic}
\end{alg}
\begin{figure}[tb]
\begin{center}
\input{figures/UTN.tex}
\caption{The tensor network of $\mat{U}$ as obtained from Algorithm~\ref{alg:matrix2TN}. Note that $r_1=r_{d+1}=1$.}
\label{fig:UiTN}
\end{center}
\end{figure}
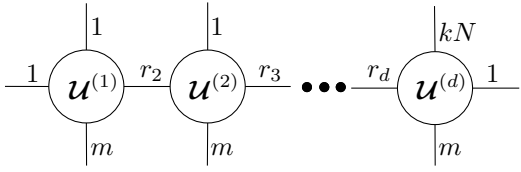
Converting the result of Algorithm~\ref{alg:matrix2TN} into the tensor network of the $km^d \times N$ matrix $\mat{U}_{0|k-1}$ is very straightforward. This is achieved through the following reshaping
\begin{align}
\label{eqn:reshapeUd}
\textrm{reshape}(\ten{U}^{(d)},[r_d,mk,N,1]).
\end{align}
Another interesting feature is that it is possible to derive explicit upper bounds for the tensor network ranks of the $\mat{U}_{0|k-1}$ matrix.
\begin{theorem}
\label{theo:rankUTN}
The tensor network ranks of the $\mat{U}_{0|k-1}$ matrix have upper bounds
\begin{align}
\label{eqn:rankU}
r_i &\leq  {i-1+m-1 \choose m-1} \; \textrm{ for } i=1,\ldots,d.
\end{align}
\end{theorem}
\begin{proof}
For $i=1$ we have that $r_1=1$, which is trivially true. Consider $j=1$ and line 3 in Algorithm~\ref{alg:matrix2TN}. For this case $\tilde{\mat{U}}$ has ${1+m-1 \choose m-1}=m$ distinct columns, which sets the upper bound for $r_2$ as derived by the SVD in line 7 to $m$. For $j=2$,  $\tilde{\mat{U}} \odot \tilde{\mat{U}}$ has ${2+m-1 \choose m-1}$ distinct columns, which similarly acts as an upper bound for $r_3$. Note that the previous ranks $r_1,r_2$ remain unchanged in any further iterations. Using this argument up to $j=d-1$ then results in the rank upper bounds~\eqref{eqn:rankU}.
% Reshaping the last core as described in~\eqref{eqn:reshapeUd} does not change $r_d$ and therefore the tensor network of $\mat{U}_{0|k-1}$ has the same ranks.
\end{proof}
In practice, when the inputs are persistent exciting, these upper bounds are always attained. Observe that the tensor network ranks only depend on the number of inputs $m-1$ and the total degree $d$. Neither the number of outputs $p$, nor the number of columns $N$ of $\mat{U}_{0|k-1}$ affect the ranks. This is completely due to the block Hankel and repeated Khatri-Rao product structures. The following example compares the tensor network ranks of $\mat{U}_{0|k-1}$ with the conventional upper bounds of Theorem~\ref{theo:MPTranks}.
\begin{example}
Consider a single-input system $(m=2)$ with $d=10$, the tensor network ranks of $\mat{U}_{0|k-1}$ are then simply $r_2=2, r_3=3, \ldots, r_{10}=10$. The conventional upper bounds are $r_2=2,r_3=4,\ldots,r_{10}=512$. Note the difference of one order of magnitude for $r_{10}$. For a system with four inputs $(m=5)$ and $d=10$, this difference becomes even larger as we have that $r_{10}=715$, compared to the conventional upper bound of 1953125.
\end{example}

\subsection{Computation of $\mat{L}_{11},\mat{L}_{21},\mat{L}_{22}$}
A major advantage of both the conventional N4SID and MOESP methods is that the orthogonal factors in the LQ decomposition never need to be computed. The tensor network modification of Algorithm~\ref{alg:moesp}, however, requires the explicit computation of the orthogonal factors $\mat{Q}_1,\mat{Q}_2$. In fact, the LQ decomposition in line 1 of Algorithm~\ref{alg:moesp} cannot be computed in tensor network form. Instead, an economical SVD of $\mat{U}_{0|k-1}$ 
\begin{align}
\label{eqn:Usvd}
\mat{U}_{0|k-1} &= \mat{W}  \mat{T} \mat{Q}^T = \begin{pmatrix}\mat{W}_1 & \mat{W}_2 \end{pmatrix}  \begin{pmatrix} \mat{T}_1 & \mat{0} \\ \mat{0} & \mat{0} \end{pmatrix}  \begin{pmatrix} \mat{Q}_1^T \\ \mat{Q}_2^T \end{pmatrix}
\end{align}
is computed with $\mat{W} \in \mathbb{R}^{km^d \times N}$ an orthogonal matrix, $\mat{T} \in \mathbb{R}^{N \times N}$ a diagonal matrix and $\mat{Q}\in \mathbb{R}^{N \times N}$ an orthogonal matrix. From assumption 2 it follows that $\mat{T}_1 \in \mathbb{R}^{r \times r}$. Let $\ten{U}^{(1)},\ldots,\ten{U}^{(d)}$ denote the tensor network of $\mat{U}_{0|k-1}$ and $\ten{W}^{(1)},\ldots,\ten{W}^{(d)}$ denote the tensor network of the orthogonal $\mat{W}$ factor in~\eqref{eqn:Usvd}. The orthogonal $\mat{W}$ factor is then computed in tensor network form using Algorithm~\ref{alg:TNSVD}.

The main idea of Algorithm~\ref{alg:TNSVD} is the orthogonalization of each of the tensors $\ten{U}^{(1)},\ldots,\ten{U}^{(d-1)}$ through a thin QR decomposition at line 3. This orthogonalization ensures that the obtained tensor network for $\mat{W}$ has the property $\mat{W}^T\mat{W}=\mat{I}_N$. The thin QR decomposition implies that the orthogonal matrix $\mat{W}$ has size $r_im \times r_{i+1}$ and $\mat{R}_i \in \mathbb{R}^{r_{i+1} \times r_{i+1}}$. The $\mat{R}_i$ factor is always absorbed by the next tensor in the network at line 5. Finally, the last tensor $\ten{U}^{(d)}$ is reshaped into an $r_dkm \times N$ matrix $\mat{U}_d$ and an economical SVD is computed at lines 7 and 8, respectively. The computed $\mat{T}$ and $\mat{Q}$ matrices are the desired factors. The computationally dominating step is the SVD of $\mat{U}_d$, which needs approximately $O(r_dmkN^2)$ flops. In order for the matrices $\mat{T},\mat{Q}$ to have the correct dimensions, it is required that $r_dkm \geq N$. Using Theorem~\ref{theo:rankUTN} it can be shown that when $L \approx L+1$ and $p-1 \leq (d-1) {d+m-2 \choose m-2}$, this condition is always satisfied. The implementation of Algorithm~\ref{alg:TNSVD} is quite straightforward. Note that due to line 8 of Algorithm~\ref{alg:matrix2TN} the first $d-1$ tensors of $\mat{U}_{0|k-1}$ are already orthogonal. This implies that the orthogonalization through the QR decompositions can be skipped and only lines 7 to 9 of Algorithm~\ref{alg:TNSVD} need to be executed on the result of Algorithm~\ref{alg:matrix2TN}.
\begin{alg}Economical SVD of $\mat{U}_{0|k-1}$ in tensor network form.\\
\textit{\textbf{Input}}: tensor network $\ten{U}^{(1)},\ldots,\ten{U}^{(d)}$ of $\mat{U}_{0|k-1}$.\\
\textit{\textbf{Output}}:\makebox[0pt][l]{ tensor network $\ten{W}^{(1)},\ldots,\ten{W}^{(d)}$ of orthogonal $\mat{W}$, }\\
\makebox[0pt][l]{\quad \quad \quad \, diagonal matrix $\mat{T}$ and orthogonal matrix $\mat{Q}$.}  
\begin{algorithmic}[1]
\For {i=1,\ldots,d-1}
\State $\mat{U}_i \gets \textrm{reshape}(\ten{U}^{(i)},[r_im,r_{i+1}])$.
\State $[\mat{W}_i,\, \mat{R}_i] \gets \textrm{QR}(\mat{U}_i)$.
\State $\ten{W}^{(i)} \gets \textrm{reshape}(\mat{W}_i,[r_i,m,1,r_{i+1}])$.
\State $\ten{U}^{(i+1)} \gets \ten{U}^{(i+1)} \times_1 \mat{R}_i$.
\EndFor
\State $\mat{U}_d \gets \textrm{reshape}(\ten{U}^{(d)},[r_dmk, N])$.
\State $[\mat{W},\,\mat{T},\,\mat{Q}] \gets \textrm{SVD}(\mat{U}_d)$.
\State $\ten{W}^{(d)}\gets \textrm{reshape}(\mat{W},[r_d,mk,N,1])$. 
\end{algorithmic}
\label{alg:TNSVD}
\end{alg}
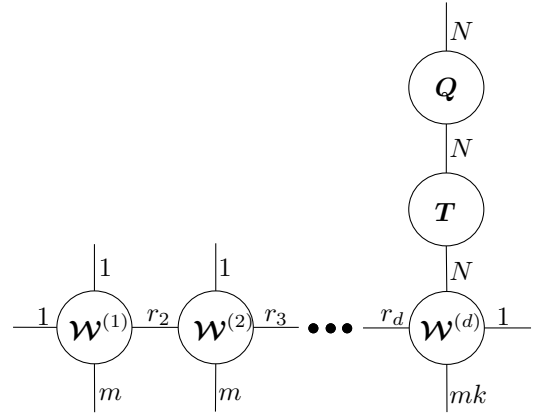
\begin{figure}[tb]
\begin{center}
\input{figures/TNSVD.tex}
\caption{The tensor network of $\mat{W}$ and matrices $\mat{T},\mat{Q}$ as obtained from Algorithm~\ref{alg:TNSVD} .}
\label{fig:TNSVD}
\end{center}
\end{figure}
A graphical representation of the tensor network for $\mat{W}$ and matrices $\mat{T},\mat{Q}$ with all dimensions labeled as obtained through Algorithm~\ref{alg:TNSVD} is shown in Figure~\ref{fig:TNSVD}. Also note that the persistence of excitation of the input can be numerically verified by inspecting the singular values on the diagonal of $\mat{T}$. Indeed, if Assumption 2 is valid, then the numerical rank of $\mat{U}_{0|k-1}$ is well-defined and the rank-gap \mbox{$\mat{T}(r,r)/\mat{T}(r+1,r+1)$} should be several orders of magnitude large.

The required matrix factors $\mat{L}_{11},\mat{L}_{21},\mat{L}_{22}$ can now be computed as
\begin{align}
\label{eqn:lfactors}
\nonumber \mat{L}_{11} &= \mat{W}_1 \mat{T}_1 \in \mathbb{R}^{km^d \times r},\\
\nonumber \mat{L}_{21} &= \mat{Y}_{0|k-1}\mat{Q}_1 \in \mathbb{R}^{kp \times r},\\
\mat{L}_{22} &= \mat{Y}_{0|k-1}\mat{Q}_2 \in \mathbb{R}^{kp \times kp}.
\end{align}
The tensor network of $\mat{L}_{11}$ is easily found as the first $d-1$ tensors are identical to the tensors of $\mat{W}$, while the $d$th tensor of $\mat{L}_{11}$ is $\ten{W}^{(d)} \times_3 \begin{pmatrix} \mat{T}_1 & \mat{O} \end{pmatrix}^T$.

\subsection{Computation of the matrices $\mat{A}$ and $\mat{C}$}
Once the matrix factors $\mat{L}_{11},\mat{L}_{21},\mat{L}_{22}$ are computed, then the conventional MOESP algorithm can be used to find $\mat{A}$ and $\mat{C}$. The SVD of the $kp \times kp$ matrix $\mat{L}_{22}$ requires approximately $O(k^3p^3)$ flops and reveals the system order $n$. The extended observability matrix $\mat{O}_k$ is then computed as $\mat{U}_1\mat{S}_1^{1/2}$, from which the first $p$ rows are taken to be the $\mat{C}$ matrix. The $\mat{A}$ matrix is found from exploiting the shift property of the extended observability matrix. In order for the pseudoinverse of the matrix $\mat{O}(1:kp-p,:)$ in line 6 of Algorithm~\ref{alg:moesp} to be unique it is required that $kp-p\geq n$, which implies that $k\geq n+1$ for the minimal case $p=1$. This means that $k$ determines the maximal system order $n$ that can be found. Computing the psuedoinverse of the $(kp-p) \times n$ matrix $\mat{O}(1:kp-p,:)$ requires approximately $O((kp-p)n^2)$ flops.

\subsection{Computation of the $\mat{U}_2^T\mat{L}_{21}\mat{L}_{11}^{-1}$ tensor network}
Line 8 of Algorithm~\ref{alg:moesp} requires the computation and partitioning of the $(kp-n) \times km^d$ matrix $\mat{U}_2^T\mat{L}_{21}\mat{L}_{11}^{-1}$. This requires the computation of the left inverse of $\mat{L}_{11}$ in tensor network form. Fortunately, from~\eqref{eqn:lfactors} it follows that $\mat{L}_{11}^{-1} = \mat{T}_1^{-1}\, \mat{W}_1^T$, as $\mat{W}_1^T\mat{W}_1=\mat{I}_r$. The transpose of $\mat{W}_1$ as a tensor network is done by permuting the second with the third dimension of each tensor in the network. The tensor network of $\mat{L}_{11}^{-1}$ is therefore obtained by permuting each of the tensors $\ten{W}^{(i)}$ into $\tilde{\ten{W}}^{(i)}$ and computing $\tilde{\ten{W}}^{(d)} \times_2 \begin{pmatrix}\mat{T}_1^{-1} & \mat{0} \end{pmatrix}$, where the inverse of $\mat{T}_1$ is obtained by inverting its diagonal. Once the tensor network of $\mat{L}_{11}^{-1}$ is obtained, multiplication with $\mat{U}_2^T\mat{L}_{21}$ is also performed on the $d$th tensor. In fact, the previous multiplication with $\mat{T}_1^{-1}$ can be combined with $\mat{U}_2^T\mat{L}_{21}$. This leads to the following theorem.
\begin{theorem}
\label{theo:MTN}
Let $\tilde{\ten{W}}^{(i)}$  $(i=1,\ldots,d)$ be the tensors $\ten{W}^{(i)}$ $(i=1,\ldots,d)$ obtained from Algorithm~\ref{alg:TNSVD} with their second and third dimensions permuted. Then the tensor network $\ten{M}^{(1)},\ldots,\ten{M}^{(d)}$ corresponding with the matrix $\mat{U}_2^T\mat{L}_{21}\mat{L}_{11}^{-1}$ is
\begin{align*}
\ten{M}^{(i)} &:= \tilde{\ten{W}}^{(i)} \in \mathbb{R}^{r_i \times 1 \times m  \times r_{i+1}} \, (i=1,\ldots,d-1),\\
\ten{M}^{(d)} &:= \tilde{\ten{W}}^{(d)} \times_2 \mat{U}_2^T\mat{L}_{21}\begin{pmatrix}\mat{T}_1^{-1} & \mat{0} \end{pmatrix} \in \mathbb{R}^{r_d \times (kp-n) \times km \times 1}.
\end{align*}
The final partitioning of $\mat{U}_2^T\mat{L}_{21}\mat{L}_{11}^{-1}$ into $k$ blocks of size $(kp-n) \times m^d$ is then obtained from
\begin{align*}
\textrm{reshape}&(\ten{M}^{(d)},[r_d,(kp-n),m,k,1]),\\
\textrm{permute}&(\ten{M}^{(d)},[1,2,4,3,5]),\\
\textrm{reshape}&(\ten{M}^{(d)},[r_d,(kp-n)k,m,1]).
\end{align*}
\end{theorem}
The final three steps to obtain the desired partitioning is due to the index of the third dimension of $\ten{M}^{(d)}$ being a grouped index $[ij]$ with $1\leq i \leq m$ and $1\leq j \leq k$. The first reshape operation separates this grouped index into its two components $i$ and $j$, after which they are permuted and the $j$ index is finally ``absorbed'' into the index of the second dimension.

\subsection{Computation of the $\mat{B},\mat{D}$ tensor network}
In order to estimate the matrices $\mat{B},\mat{D}$ the linear system~\eqref{eqn:BDsystem} needs to be solved. Computing the pseudoinverse of the $k(kp-n) \times (p+n)$ matrix on the left hand side of~\eqref{eqn:BDsystem} requires approximately $O(k(kp-n)(p+n)^2)$ flops. If we denote this pseudoinverse by $\mat{L}^{-1}$, then the concatenation of $\mat{D}$ with $\mat{B}$ is found as
\begin{align}
\label{eqn:BDsol}
\begin{pmatrix} \mat{D} \\ \mat{B} \end{pmatrix} &= \mat{L}^{-1}\,\begin{pmatrix}\mat{M}_1 \\ \mat{M}_2 \\ \vdots \\ \mat{M}_{k-1} \\ \mat{M}_k \end{pmatrix}.
\end{align}
The partitioned $\mat{U}_2^T\mat{L}_{21}\mat{L}_{11}^{-1}$ matrix is already available to us as a tensor network from the previous section. Therefore, the contraction $\ten{M}^{(d)} \times_2 \mat{L}^{-1}$ results in the tensor network that represents the concatenation of $\mat{D}$ with $\mat{B}$. As mentioned in Section~\ref{sec:pss}, the affine terms of both $\mat{B}$ and $\mat{D}$ are defined to be zero and hence need to be set explicitly to zero in the estimation. This can be achieved by multiplying \eqref{eqn:BDsol} to the right with the $m^d \times m^d$ matrix
\begin{align*}
\mat{P} &=
\begin{pmatrix} 0 & \mat{0}\\
\mat{0} & \mat{I}
\end{pmatrix},
\end{align*}
which is essentially the unit matrix with entry $(1,1)$ set to zero. The matrix $\mat{P}$ has the following exact uniform rank-2 tensor network representation
\begin{align*}
\mat{P} &= \mat{I}_m \otimes \mat{I}_m \otimes \cdots \otimes \mat{I}_m \otimes \mat{I}_m +\\
& \begin{pmatrix}\mat{e}_1 & \mat{0} \end{pmatrix} \otimes \begin{pmatrix}\mat{e}_1 & \mat{0} \end{pmatrix} \otimes \cdots \otimes  \begin{pmatrix}\mat{e}_1 & \mat{0} \end{pmatrix} \otimes  \begin{pmatrix}-\mat{e}_1 & \mat{0} \end{pmatrix},
\end{align*}
where all Kronecker factors are $m \times m$ matrices and $\mat{e}_1$ is the first canonical basis vector in $\mathbb{R}^m$. The multiplication of $\begin{pmatrix} \mat{D}^T & \mat{B}^T \end{pmatrix}^T\mat{P}$ in tensor network form is then achieved by contracting all $d$ corresponding tensors with a total computational complexity of approximately $O(dr^2m^2)$ flops, where $r$ denotes the maximal tensor network rank.

\subsection{Tensor network MOESP}
Algorithm~\ref{alg:TNmoesp} describes our tensor network MOESP (TNMOESP) method for the identification of polynomial state space systems~\eqref{eqn:pss}. Again, the fact that both $\mat{A},\mat{C}$ and $\mat{B},\mat{D}$ are estimated separately in MOESP is an advantage for the required modifications to construct a tensor network version of this algorithm. Algorithm~\ref{alg:TNmoesp} is implemented in the TNMOESP MATLAB package as TNmoesp.m.
\begin{alg}Tensor network MOESP algorithm\\
\label{alg:TNmoesp}
\textit{\textbf{Input}}: $L$ samples $(\mat{u}_0,\mat{y}_0),\ldots,(\mat{u}_{L-1},\mat{y}_{L-1})$, $k$\\
\textit{\textbf{Output}}:\makebox[0pt][l]{ Matrices $\mat{A},\mat{C}$, tensor network $\ten{T}^{(1)},\ldots,\ten{T}^{(d)}$.}
\begin{algorithmic}[1]
\State $\ten{U}^{(1)},\ldots,\ten{U}^{(d)} \gets$ Algorithm~\ref{alg:matrix2TN}
\State $\textrm{reshape}(\ten{U}^{(d)},[r_d,mk,N,1])$
\State $\ten{W}^{(1)},\ldots,\ten{W}^{(d)}, \mat{T}, \mat{Q} \gets$ Algorithm~\ref{alg:TNSVD}
\State $\begin{pmatrix}\mat{L}_{21} & \mat{L}_{22} \end{pmatrix} = \mat{Y}_{0|k-1}\,\mat{Q}$
\State SVD of $L_{22} = \begin{pmatrix}\mat{U}_1 & \mat{U}_2 \end{pmatrix} \begin{pmatrix} \mat{S}_1 & \mat{0}\\\mat{0} & \mat{0} \end{pmatrix} \begin{pmatrix} \mat{V}_1^T \\ \mat{V}_2^T\end{pmatrix}$
\State Define system order as $n:=\textrm{rank}(L_{22})$
\State $\mat{O}_k=\mat{U}_1\,\mat{S}_1^{1/2}$ and  $\mat{C} = \mat{O}_k(1:p,:)$
\State Compute $\mat{A}$ from $\mat{O}_k(1:kp-p,:)\mat{A}=\mat{O}(p+1:kp,:)$
\State Partition $\mat{U}_2^T:=\begin{pmatrix}\mat{L}_1 & \cdots & \mat{L}_k \end{pmatrix}$ into $k$ blocks of size \mbox{$(kp-n) \times p$}.
\State $\tilde{\ten{W}}^{(i)} :=$ permute($\ten{W}^{(i)}$,[1,3,2,4]), $i=1,\ldots,d$
\State $\ten{M}^{(1)},\ldots,\ten{M}^{(d)} \gets$ Theorem~\ref{theo:MTN}
\State $\ten{T}^{(i)} := \ten{M}^{(i)},\;i=1,\ldots,d-1$
\State $\ten{T}^{(d)} := \ten{M}^{(d)} \times_2 \mat{L}^{-1}$ 
\State Contract $\ten{T}^{(1)},\ldots,\ten{T}^{(d)}$ with tensor network of matrix $\mat{P}$
\end{algorithmic}
\end{alg}

\section{Simulation of the polynomial state space model}
\label{sec:sim}
Algorithm~\ref{alg:TNmoesp} does not return separate tensor networks for $\mat{B}$ and $\mat{D}$. This is also not strictly required when simulating the model. In fact, having the concatenation of $\mat{B}$ and $\mat{D}$ available in one tensor network simplifies the simulation. Instead of forming the repeated Kronecker product $\mat{u}_t\kpr{d}$, one simply needs to contract $\mat{u}_t$ with each tensor in the network as indicated in Figure~\ref{fig:sim}. These contractions are lines 4 and 7 in Algorithm~\ref{alg:simTN}. The computational complexity of Algorithm~\ref{alg:simTN} is approximately $O(dmr^2+n^2)$ flops, where $r$ is the maximal tensor network rank. After the final contraction with $\mat{u}_t$ on line 7, we obtain the $(p+n) \times 1$ vector that is the concatenation of $\mat{D}\mat{u}_t\kpr{d}$ with $\mat{B}\mat{u}_t\kpr{d}$.
 Algorithm~\ref{alg:simTN} is implemented in the TNMOESP MATLAB package as simTNss.m.
\begin{alg}Simulation of~\eqref{eqn:pss} in tensor network form\\
\label{alg:simTN}
\textit{\textbf{Input}}: Initial state vector $\mat{x}_0$, inputs $\mat{u}_0,\ldots,\mat{u}_{L-1}$, matrices $\mat{A},\mat{C}$ and tensor network $\ten{T}^{(1)},\ldots,\ten{T}^{(d)}$ of $\begin{pmatrix} \mat{D}^T & \mat{B}^T \end{pmatrix}^T$\\
\textit{\textbf{Output}}:\makebox[0pt][l]{ outputs $\mat{y}_0,\ldots,\mat{y}_{L-1}$.}
\begin{algorithmic}[1]
\For{i=0,\ldots,L-1}
\State $\mat{z} = \ten{T}^{(1)}$
\For{j=1,\ldots,d-1}
\State $\mat{z} = \mat{u}_i^T\, \textrm{reshape}(\mat{z},[m,r_{j+1}])$
\State $\mat{z} = \mat{z}\, \textrm{reshape}(\ten{T}^{(j+1)},[r_{j+1},mr_{j+2}])$
\EndFor
\State $\mat{z} = \textrm{reshape}(\mat{z},[p+n,m])\,\mat{u}_i^T$
\State $\mat{y}_i = \mat{C}\mat{x}_0 + \mat{z}(1:p)$
\State $\mat{x}_0 = \mat{A}\mat{x}_0 + \mat{z}(p+1:p+n)$
\EndFor
\end{algorithmic}
\end{alg}
\begin{figure}[tb]
\begin{center}
\input{figures/TNsim.tex}
\caption{Contraction of the $\ten{T}^{(1)},\ldots,\ten{T}^{(d)}$ tensor network with a vector $\mat{u}_t$ to obtain the concatenation of $\mat{D}\mat{u}_t\kpr{d}$ with $\mat{B}\mat{u}_t\kpr{d}$.}
\label{fig:sim}
\end{center}
\end{figure}
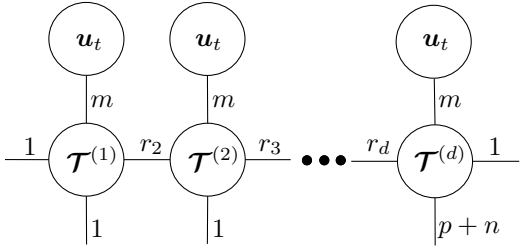
\section{Numerical Experiments}
\label{sec:experiments}
In this section we demonstrate the efficacy of TNMOESP and compare its performance with other state-of-the-art nonlinear system identification methods. All algorithms were implemented in MATLAB and run on a desktop computer with 8 cores running at 3.4 GHz and 64 GB RAM. The TNMOESP package can be freely downloaded from \url{https://github.com/kbatseli/TNMOESP}.

\subsection{Verifying correctness of TNMOESP}
\label{subsec:ex1}
First, we verify whether TNMOESP is able to correctly recover the polynomial state space model and compare its performance with the matrix-based implementation Algorithm~\ref{alg:moesp}. We fix the values $n=5$, $m=5$, $p=3$ and $L=2048$ and construct polynomial state space models~\eqref{eqn:pss} for degrees $d=2$ up to $d=8$. All $5 \times 5$ $\mat{A}$ matrices are constructed such that all eigenvalues have absolute values strictly smaller than 1, thus ensuring stability of the model. All other coefficients of $\mat{B},\mat{C},\mat{D}$ and the input signals were chosen from a standard normal distribution. The simulated outputs for the constructed models were then used to identify the models with both TNMOESP and a matrix based implementation of Algorithm~\ref{alg:moesp}\footnote{This matrix based implementation is available in the TNMOESP package as moespd.m.}.

The identified models were then validated by applying 1024 samples of different standard normal distributed inputs and comparing the simulated outputs of the ``real'' system with the outputs of the identified system. Table~\ref{tbl:ex1} lists the run times in seconds for the system identification algorithms to finish and relative errors $||\mat{y}-\hat{\mat{y}}||_F/||\mat{y}||_F$, where $\mat{y}$ is the real output and $\hat{\mat{y}}$ is the output computed from the estimated models. The matrix based implementation was not able to estimate the $d=8$ model due to insufficient memory. TNMOESP consistently outperforms the matrix based method and for $d=7$ is about 20 times faster than the matrix based method. The relative validation errors indicate that the models were estimated accurately up to machine precision, thus validating the correctness of TNMOESP.
\begin{table}[ht]
\begin{center}
\caption{Total run times for identification and relative validation errors for increasing $d$.}
\begin{tabular}{@{}rrrrr@{}}
$d$ & \multicolumn{2}{c}{Total Run time [s]} & \multicolumn{2}{c}{Rel. Val. error} \\ \midrule
& Algorithm~\ref{alg:moesp} & TNMOESP & Algorithm~\ref{alg:moesp} & TNMOESP\\\midrule
$2$& $5.2$ & $3.8$ & \num{1.1e-15} & \num{1.2e-14}\\
$3$& $7.1$& $4.4$ & \num{9.2e-16} & \num{4.7e-14}\\
$4$& $14.0$ & $7.2$ & \num{2.2e-15} & \num{3.1e-14} \\
$5$& $27.4$ & $7.7$ & \num{8.3e-15} & \num{2.9e-14} \\
$6$& $76.3$ & $9.8$ & \num{1.2e-14} & \num{1.4e-14}\\
$7$& $258$ & $13.2$ & \num{8.0e-14} & \num{1.3e-13} \\
$8$& NA & $15.1$ & NA & \num{4.4e-13}\\
\end{tabular}
\label{tbl:ex1}
\end{center}
\end{table}

For each of the constructed models with $n=5$, $m=5$, $p=3$, 5000 output samples were computed using a standard Kronecker product implementation of \eqref{eqn:pssmatricized} and with Algorithm~\ref{alg:simTN}. The total simulation times for both methods are listed in Table~\ref{tab:ex1_2}. The benefit of doing the simulation of the model with tensor networks becomes more pronounced as the degree $d$ increases with Algorithm~\ref{alg:simTN} being 70 times faster than the standard implementation.
\begin{table}[tb]
\begin{center}
\caption{Total run times for computation of 5000 output samples with Kronecker products and with Algorithm~\ref{alg:simTN} for increasing $d$.}
\label{tab:ex1_2}	
\begin{tabular}{@{}lccccccc@{}}
$d$		 & 2 & 3 & 4 & 5 & 6 & 7 & 8\\
\midrule
Kron. [s] & $0.07$ & $0.13$&$0.19$ &$0.35$ &  $1.44$ & $5.53$ &$28.3$\\
Alg.~\ref{alg:simTN} [s] & $0.07$ & $0.09$&$0.12$ &$0.16$ & $0.23$ & $0.31$&$0.40$ \\
\midrule
\end{tabular}
\end{center}
\end{table}

\subsection{Influence of noise - output error model}
In this experiment the effect of noise on the measured output on the identification with TNMOESP is investigated. Noise on the output implies that the block Hankel matrix $\mat{Y}$ will be perturbed by a block Hankel noise matrix $\mat{E}$, and therefore all singular values of $\mat{L}_{22}$ will be in the worst case perturbed by $||\mat{E}||_2$. This needs to be taken into account when estimating the system order $n$. Luckily, we are only interested in the left singular vectors of $\mat{L}_{22}$, which are not very sensitive to perturbations when $kp < N$~\cite[p.~166]{katayama2005subspace}. A polynomial state space system~\eqref{eqn:pss} was constructed as in Experiment~\ref{subsec:ex1} with $m=5,p=3,n=5$ and $d=5$. The outputs were simulated by exciting the system with 4096 normal distributed input samples. Six separate data sets with signal-to-noise ratios (SNRs) of 5dB, 10dB, 15dB, 20dB, 25dB and 30dB, respectively, were then created by adding zero-mean Gaussian noise to the simulated outputs. These six data sets were then used with TNMOESP to estimate a polynomial state space model. A different set of 4096 input samples was then used to generate validation data on the ``real" and estimated models. We define the simulation SNR (sim SNR) as
\begin{align*}
10\,\textrm{log}_{10}\, \left(\frac{\sum_i \mat{y}_i^2}{\sum_i ( \mat{y}_i-\hat{ \mat{y}}_i)^2}  \right)
\end{align*}
where $ \mat{y}_i$ is the output validation signal uncorrupted by noise and~$\hat{ \mat{y}}_i$ is the simulated output from the estimated model. Table~\ref{table:SNR} compares the SNR of the signals used in the identification (ID SNR) with the SNR of the simulated signal (SIM SNR). The relative validation errors $||\mat{y}-\hat{\mat{y}}||_F/||\mat{y}||_F$ are also indicated. As expected, the identification results improve when data of increasing SNR is used, which is indicated by the monotonically decreasing relative validation error. The simulated output of the estimated model has a consistent SNR improvement compared to the output used for the identification. The better the quality of the signals used for identification, the smaller the improvement. Even for the 5dB case, TNMOESP is able to correctly identify the underlying model, which indicates the robustness of the algorithm with respect to noise.
\begin{table}[ht]
\centering
\caption{Identification under 6 different SNR levels.}%\vspace{3pt}
\label{table:SNR} % is used to refer this table in the text
\begin{tabular}{@{}lrrrrrr@{}}\midrule
ID SNR  & 5dB & 10dB & 15dB & 20dB & 25dB & 30dB \\\midrule
SIM SNR & 9dB & 13dB & 17dB & 21dB & 26dB & 31dB \\
Rel. Val. error & 0.35 & 0.21 & 0.14 & 0.09 & 0.05 & 0.02\\ \midrule
\end{tabular}
\end{table}

\subsection{High-end valve control amplifier}
\label{subsec:amplifier}
In this experiment we compare the performance of TNMOESP with other models and methods on real-world data. The data set is from the same experiment as described in~\cite[p.~3936]{schoukens2011parametric} and the system under consideration is a high-end valve control amplifier, which is normally used as a preamplifier for audio signals. The amplifier is a single-input-single-output system and was fed a flat spectrum random phase multisine with a period of 4096 samples, sampled at 1.25 MHz.

We compare four different models and methods. For each of these models/methods, only the one with the best relative validation error is reported. First, a linear state space system was identified by Algorithm~\ref{alg:moesp} with system order $n=3$ using the first 1000 samples. Then, a polynomial state space model was identified using TNMOESP with $d=6$ and $n=30$, also using the first 1000 samples. In addition, we also identified a Volterra model of degree $d=2$ and memory $M=30$ using the MALS tensor method described in~\cite{MVMALS}, also using the first 1000 samples. Finally, a general polynomial state space as described in~\cite{PADUART2010647} was identified using the iterative methods of the PNLSS MATLAB toolbox\footnote{The PNLSS MATLAB toolbox can be freely downloaded from \url{homepages.vub.ac.be/~ktiels/pnlss.html}} with $n=15$ and where both polynomials of state and output equations are of degree 4. In order to obtain good validation errors, the general polynomial state space model needed to be estimated on 2 periods of the input signal. Each of the models were then used to simulate the output from the input that was not used for identification. The run times and relative validation errors $||\mat{y}-\hat{\mat{y}||}/||\mat{y}||$, where $\mat{y}$ denotes the measured output and $\hat{\mat{y}}$ denotes the simulated output, for each of the methods and models are listed in Table~\ref{table:amplifier}. 
\begin{figure}[tb]
\begin{center}
\includegraphics[width=.5\textwidth]{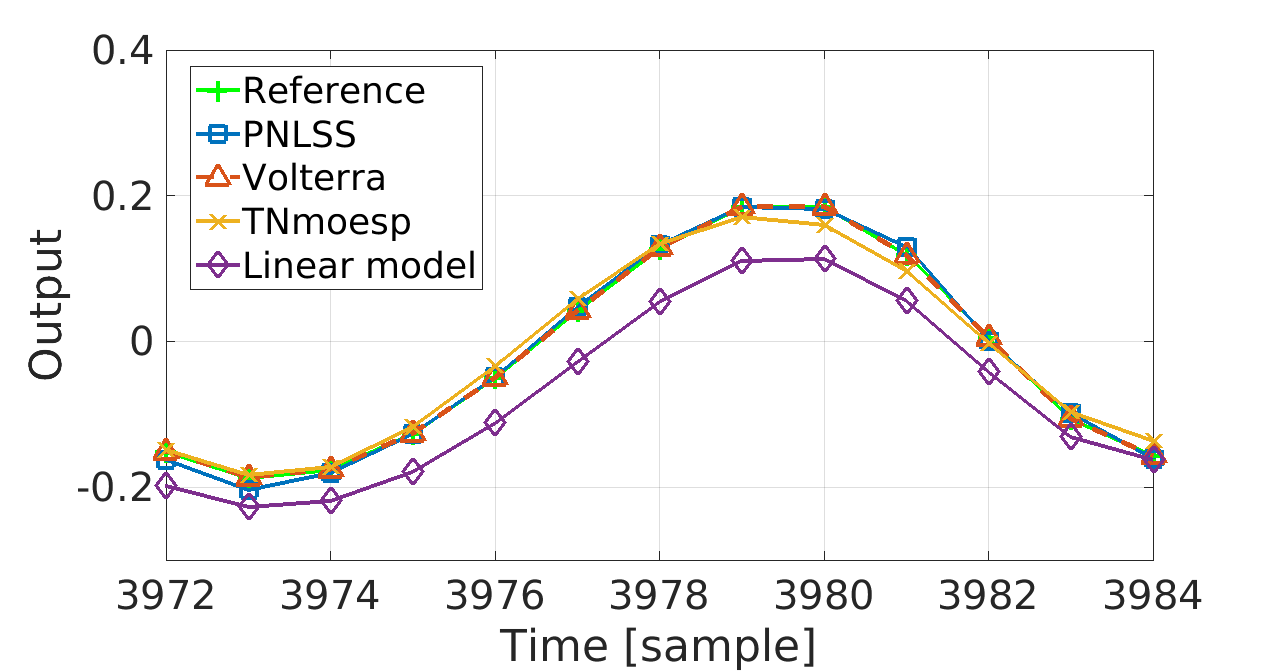}
\caption{Detail of reference and simulated amplifier output from different models.}
\label{fig:amplifier}
\end{center}
\end{figure}
\begin{table}[ht]
\centering
\caption{Run times and relative validation errors for four different models and methods.}%\vspace{3pt}
\label{table:amplifier} % is used to refer this table in the text
\begin{tabular}{@{}lrr@{}}\midrule
Method & Run time [s] & Rel. Val. error \\ \midrule
Linear & 0.26 &  0.418\\
TNMOESP &0.69 & 0.148 \\
PNLSS & 14264 & 0.087\\ 
Volterra & 1.61  & 0.004\\
\midrule
\end{tabular}
\end{table}

The linear state space model can be identified very quickly but performs the worst, while TNMOESP improves the validation at the cost of a slightly longer run time. The general polynomial state space system is capable of improving the validation error by one order of magnitude at the cost of a very significant run time. Convergence of the iterative method in the PNLSS toolbox was rather slow as it took 12317 seconds for the relative validation error to drop to 0.29. An interesting avenue of future research is to investigate whether it is possible to further refine the model obtained through TNMOESP by using it as an initial guess for the iterative routines in the PNLSS toolbox. This might alleviate the long run time due to slow convergence. Surprisingly, the Volterra model is able to achieve a relative validation error that is another order of magnitude smaller than the general polynomial state space system, which might suggest that the real-world system is better described by a Volterra model rather than a polynomial state space model. Figure~\ref{fig:amplifier} shows a few samples of the reference output and simulated outputs for the four different models. Due to the scale of the figure, it is not possible to distinguish the output from the Volterra model from the reference. As evident from the figure, all nonlinear models produce outputs that are closer to the real output compared to the linear model.

\section{Conclusions}
\label{sec:conclusions}
This article introduces a particular polynomial extension of the linear state space model and develops an efficient tensor network subspace identification method. The polynomial nonlinearity is described by tensor networks, which solves the need to store an exponentially large number of parameters. It was shown how the block Hankel input matrix that lies at the heart of subspace identification is exactly represented by a low rank tensor network, significantly reducing the computational and storage complexity during identification and simulation. The correctness of our tensor network subspace algorithm was demonstrated through numerical experiments, together with its robustness in the presence of noise on the measured output. Finally, the accuracy and total run time of our method were compared with three other models and methods. Future work includes the investigation whether models obtained through TNMOESP are good candidates as initial guesses for the iterative methods in the PNLSS toolbox.

\section*{Acknowledgements}
The authors would like to express their sincere gratitude to dr. Maarten Schoukens and dr. Koen Tiels for providing the real-world experiment data of Experiment~\ref{subsec:amplifier} to us and their support in using the PNLSS MATLAB toolbox.
\bibliographystyle{plain}        % Include this if you use bibtex 
\bibliography{references}           % and a bib file to produce the 

\end{document}

%% file: figures/TNgraphs.tex
% Graphic for TeX using PGF
% Title: /home/kim/Dropbox/Work/Papers/Journal/16.MERA/figs/TNgraphs.dia
% Creator: Dia v0.97.2
% CreationDate: Wed Apr 19 13:58:56 2017
% For: kim
% \usepackage{tikz}
% The following commands are not supported in PSTricks at present
% We define them conditionally, so when they are implemented,
% this pgf file will use them.
\ifx\du\undefined
  \newlength{\du}
\fi
\setlength{\du}{4\unitlength}
\begin{tikzpicture}
\pgftransformxscale{1.000000}
\pgftransformyscale{-1.000000}
\definecolor{dialinecolor}{rgb}{0.000000, 0.000000, 0.000000}
\pgfsetstrokecolor{dialinecolor}
\definecolor{dialinecolor}{rgb}{1.000000, 1.000000, 1.000000}
\pgfsetfillcolor{dialinecolor}
\definecolor{dialinecolor}{rgb}{1.000000, 1.000000, 1.000000}
\pgfsetfillcolor{dialinecolor}
\pgfpathellipse{\pgfpoint{-6.073446\du}{10.779091\du}}{\pgfpoint{2.900000\du}{0\du}}{\pgfpoint{0\du}{2.800000\du}}f
\pgfusepath{fill}
\pgfsetlinewidth{0.100000\du}
\pgfsetdash{}{0pt}
\pgfsetdash{}{0pt}
\definecolor{dialinecolor}{rgb}{0.000000, 0.000000, 0.000000}
\pgfsetstrokecolor{dialinecolor}
\pgfpathellipse{\pgfpoint{-6.073446\du}{10.779091\du}}{\pgfpoint{2.900000\du}{0\du}}{\pgfpoint{0\du}{2.800000\du}}
\pgfusepath{stroke}
\pgfsetlinewidth{0.100000\du}
\pgfsetdash{}{0pt}
\pgfsetdash{}{0pt}
\pgfsetbuttcap
{
\definecolor{dialinecolor}{rgb}{0.000000, 0.000000, 0.000000}
\pgfsetfillcolor{dialinecolor}
% was here!!!
\definecolor{dialinecolor}{rgb}{0.000000, 0.000000, 0.000000}
\pgfsetstrokecolor{dialinecolor}
\draw (-4.022836\du,12.758990\du)--(-4.022836\du,16.80\du);
}
\pgfsetlinewidth{0.100000\du}
\pgfsetdash{}{0pt}
\pgfsetdash{}{0pt}
\pgfsetbuttcap
{
\definecolor{dialinecolor}{rgb}{0.000000, 0.000000, 0.000000}
\pgfsetfillcolor{dialinecolor}
% was here!!!
\definecolor{dialinecolor}{rgb}{0.000000, 0.000000, 0.000000}
\pgfsetstrokecolor{dialinecolor}
\draw (-8.124055\du,12.758990\du)--(-8.124055\du,16.80\du);
}
\pgfsetlinewidth{0.100000\du}
\pgfsetdash{}{0pt}
\pgfsetdash{}{0pt}
\pgfsetbuttcap
{
\definecolor{dialinecolor}{rgb}{0.000000, 0.000000, 0.000000}
\pgfsetfillcolor{dialinecolor}
% was here!!!
\definecolor{dialinecolor}{rgb}{0.000000, 0.000000, 0.000000}
\pgfsetstrokecolor{dialinecolor}
\draw (-6.073446\du,13.579091\du)--(-6.073446\du,16.80\du);
}
% setfont left to latex
\definecolor{dialinecolor}{rgb}{0.000000, 0.000000, 0.000000}
\pgfsetstrokecolor{dialinecolor}
\node[anchor=west] at (-8.20\du,10.6\du){$\ten{A}$};
\definecolor{dialinecolor}{rgb}{1.000000, 1.000000, 1.000000}
\pgfsetfillcolor{dialinecolor}
\pgfpathellipse{\pgfpoint{-36.297515\du}{10.779091\du}}{\pgfpoint{2.900000\du}{0\du}}{\pgfpoint{0\du}{2.800000\du}}
\pgfusepath{fill}
\pgfsetlinewidth{0.100000\du}
\pgfsetdash{}{0pt}
\pgfsetdash{}{0pt}
\definecolor{dialinecolor}{rgb}{0.000000, 0.000000, 0.000000}
\pgfsetstrokecolor{dialinecolor}
\pgfpathellipse{\pgfpoint{-36.297515\du}{10.779091\du}}{\pgfpoint{2.900000\du}{0\du}}{\pgfpoint{0\du}{2.800000\du}}
\pgfusepath{stroke}
% setfont left to latex
\definecolor{dialinecolor}{rgb}{0.000000, 0.000000, 0.000000}
\pgfsetstrokecolor{dialinecolor}
\node[anchor=west] at (-37.80\du,10.80\du){$a$};
\definecolor{dialinecolor}{rgb}{1.000000, 1.000000, 1.000000}
\pgfsetfillcolor{dialinecolor}
\pgfpathellipse{\pgfpoint{-26.560105\du}{10.779091\du}}{\pgfpoint{2.900000\du}{0\du}}{\pgfpoint{0\du}{2.800000\du}}
\pgfusepath{fill}
\pgfsetlinewidth{0.100000\du}
\pgfsetdash{}{0pt}
\pgfsetdash{}{0pt}
\definecolor{dialinecolor}{rgb}{0.000000, 0.000000, 0.000000}
\pgfsetstrokecolor{dialinecolor}
\pgfpathellipse{\pgfpoint{-26.560105\du}{10.779091\du}}{\pgfpoint{2.900000\du}{0\du}}{\pgfpoint{0\du}{2.800000\du}}
\pgfusepath{stroke}
\pgfsetlinewidth{0.100000\du}
\pgfsetdash{}{0pt}
\pgfsetdash{}{0pt}
\pgfsetbuttcap
{
\definecolor{dialinecolor}{rgb}{0.000000, 0.000000, 0.000000}
\pgfsetfillcolor{dialinecolor}
% was here!!!
\definecolor{dialinecolor}{rgb}{0.000000, 0.000000, 0.000000}
\pgfsetstrokecolor{dialinecolor}
\draw (-26.560105\du,13.579091\du)--(-26.560105\du,16.80\du);
}
% setfont left to latex
\definecolor{dialinecolor}{rgb}{0.000000, 0.000000, 0.000000}
\pgfsetstrokecolor{dialinecolor}
\node[anchor=west] at (-28.20\du,10.80\du){$\mat{a}$};
\definecolor{dialinecolor}{rgb}{1.000000, 1.000000, 1.000000}
\pgfsetfillcolor{dialinecolor}
\pgfpathellipse{\pgfpoint{-16.948419\du}{10.779091\du}}{\pgfpoint{2.900000\du}{0\du}}{\pgfpoint{0\du}{2.800000\du}}
\pgfusepath{fill}
\pgfsetlinewidth{0.100000\du}
\pgfsetdash{}{0pt}
\pgfsetdash{}{0pt}
\definecolor{dialinecolor}{rgb}{0.000000, 0.000000, 0.000000}
\pgfsetstrokecolor{dialinecolor}
\pgfpathellipse{\pgfpoint{-16.948419\du}{10.779091\du}}{\pgfpoint{2.900000\du}{0\du}}{\pgfpoint{0\du}{2.800000\du}}
\pgfusepath{stroke}
\pgfsetlinewidth{0.100000\du}
\pgfsetdash{}{0pt}
\pgfsetdash{}{0pt}
\pgfsetbuttcap
{
\definecolor{dialinecolor}{rgb}{0.000000, 0.000000, 0.000000}
\pgfsetfillcolor{dialinecolor}
% was here!!!
\definecolor{dialinecolor}{rgb}{0.000000, 0.000000, 0.000000}
\pgfsetstrokecolor{dialinecolor}
\draw (-14.897810\du,12.758990\du)--(-14.897810\du,16.80\du);
}
\pgfsetlinewidth{0.100000\du}
\pgfsetdash{}{0pt}
\pgfsetdash{}{0pt}
\pgfsetbuttcap
{
\definecolor{dialinecolor}{rgb}{0.000000, 0.000000, 0.000000}
\pgfsetfillcolor{dialinecolor}
% was here!!!
\definecolor{dialinecolor}{rgb}{0.000000, 0.000000, 0.000000}
\pgfsetstrokecolor{dialinecolor}
\draw (-18.999029\du,12.758990\du)--(-18.999029\du,16.80\du);
}
% setfont left to latex
\definecolor{dialinecolor}{rgb}{0.000000, 0.000000, 0.000000}
\pgfsetstrokecolor{dialinecolor}
\node[anchor=west] at (-19.00\du,10.6\du){$\mat{A}$};
\end{tikzpicture}

%% file: figures/MPO.tex
% Graphic for TeX using PGF
% Title: /home/kim/Dropbox/Work/Papers/Journal/19.TNMoesp/figures/MPO.dia
% Creator: Dia v0.97.2
% CreationDate: Tue Sep 12 10:42:51 2017
% For: kim
% \usepackage{tikz}
% The following commands are not supported in PSTricks at present
% We define them conditionally, so when they are implemented,
% this pgf file will use them.
\ifx\du\undefined
  \newlength{\du}
\fi
\setlength{\du}{4\unitlength}
\begin{tikzpicture}
\pgftransformxscale{1.000000}
\pgftransformyscale{-1.000000}
\definecolor{dialinecolor}{rgb}{0.000000, 0.000000, 0.000000}
\pgfsetstrokecolor{dialinecolor}
\definecolor{dialinecolor}{rgb}{1.000000, 1.000000, 1.000000}
\pgfsetfillcolor{dialinecolor}
\definecolor{dialinecolor}{rgb}{1.000000, 1.000000, 1.000000}
\pgfsetfillcolor{dialinecolor}
\pgfpathellipse{\pgfpoint{43.590000\du}{6.570000\du}}{\pgfpoint{3.550000\du}{0\du}}{\pgfpoint{0\du}{3.450000\du}}
\pgfusepath{fill}
\pgfsetlinewidth{0.100000\du}
\pgfsetdash{}{0pt}
\pgfsetdash{}{0pt}
\definecolor{dialinecolor}{rgb}{0.000000, 0.000000, 0.000000}
\pgfsetstrokecolor{dialinecolor}
\pgfpathellipse{\pgfpoint{43.590000\du}{6.570000\du}}{\pgfpoint{3.550000\du}{0\du}}{\pgfpoint{0\du}{3.450000\du}}
\pgfusepath{stroke}
\pgfsetlinewidth{0.100000\du}
\pgfsetdash{}{0pt}
\pgfsetdash{}{0pt}
\pgfsetbuttcap
{
\definecolor{dialinecolor}{rgb}{0.000000, 0.000000, 0.000000}
\pgfsetfillcolor{dialinecolor}
% was here!!!
\definecolor{dialinecolor}{rgb}{0.000000, 0.000000, 0.000000}
\pgfsetstrokecolor{dialinecolor}
\draw (47.140000\du,6.570000\du)--(51.513200\du,6.570000\du);
}
\pgfsetlinewidth{0.100000\du}
\pgfsetdash{}{0pt}
\pgfsetdash{}{0pt}
\pgfsetbuttcap
{
\definecolor{dialinecolor}{rgb}{0.000000, 0.000000, 0.000000}
\pgfsetfillcolor{dialinecolor}
% was here!!!
\definecolor{dialinecolor}{rgb}{0.000000, 0.000000, 0.000000}
\pgfsetstrokecolor{dialinecolor}
\draw (35.895000\du,6.560000\du)--(40.040000\du,6.570000\du);
}
\pgfsetlinewidth{0.100000\du}
\pgfsetdash{}{0pt}
\pgfsetdash{}{0pt}
\pgfsetbuttcap
{
\definecolor{dialinecolor}{rgb}{0.000000, 0.000000, 0.000000}
\pgfsetfillcolor{dialinecolor}
% was here!!!
\definecolor{dialinecolor}{rgb}{0.000000, 0.000000, 0.000000}
\pgfsetstrokecolor{dialinecolor}
\draw (41.079800\du,9.009500\du)--(41.079800\du,13.00\du);
}
\pgfsetlinewidth{0.100000\du}
\pgfsetdash{}{0pt}
\pgfsetdash{}{0pt}
\pgfsetbuttcap
{
\definecolor{dialinecolor}{rgb}{0.000000, 0.000000, 0.000000}
\pgfsetfillcolor{dialinecolor}
% was here!!!
\definecolor{dialinecolor}{rgb}{0.000000, 0.000000, 0.000000}
\pgfsetstrokecolor{dialinecolor}
\draw (46.100200\du,9.009500\du)--(46.100200\du,13.00\du);
}
% setfont left to latex
\definecolor{dialinecolor}{rgb}{0.000000, 0.000000, 0.000000}
\pgfsetstrokecolor{dialinecolor}
\node[anchor=west] at (39.6\du,14.50\du){$i_1$};
% setfont left to latex
\definecolor{dialinecolor}{rgb}{0.000000, 0.000000, 0.000000}
\pgfsetstrokecolor{dialinecolor}
\node[anchor=west] at (44.50\du,14.50\du){$j_1$};
\definecolor{dialinecolor}{rgb}{1.000000, 1.000000, 1.000000}
\pgfsetfillcolor{dialinecolor}
\pgfpathellipse{\pgfpoint{55.063200\du}{6.570000\du}}{\pgfpoint{3.550000\du}{0\du}}{\pgfpoint{0\du}{3.450000\du}}
\pgfusepath{fill}
\pgfsetlinewidth{0.100000\du}
\pgfsetdash{}{0pt}
\pgfsetdash{}{0pt}
\definecolor{dialinecolor}{rgb}{0.000000, 0.000000, 0.000000}
\pgfsetstrokecolor{dialinecolor}
\pgfpathellipse{\pgfpoint{55.063200\du}{6.570000\du}}{\pgfpoint{3.550000\du}{0\du}}{\pgfpoint{0\du}{3.450000\du}}
\pgfusepath{stroke}
\pgfsetlinewidth{0.100000\du}
\pgfsetdash{}{0pt}
\pgfsetdash{}{0pt}
\pgfsetbuttcap
{
\definecolor{dialinecolor}{rgb}{0.000000, 0.000000, 0.000000}
\pgfsetfillcolor{dialinecolor}
% was here!!!
\definecolor{dialinecolor}{rgb}{0.000000, 0.000000, 0.000000}
\pgfsetstrokecolor{dialinecolor}
\draw (52.553000\du,9.009500\du)--(52.553000\du,13.00\du);
}
\pgfsetlinewidth{0.100000\du}
\pgfsetdash{}{0pt}
\pgfsetdash{}{0pt}
\pgfsetbuttcap
{
\definecolor{dialinecolor}{rgb}{0.000000, 0.000000, 0.000000}
\pgfsetfillcolor{dialinecolor}
% was here!!!
\definecolor{dialinecolor}{rgb}{0.000000, 0.000000, 0.000000}
\pgfsetstrokecolor{dialinecolor}
\draw (57.573400\du,9.009500\du)--(57.573400\du,13.00\du);
}
% setfont left to latex
\definecolor{dialinecolor}{rgb}{0.000000, 0.000000, 0.000000}
\pgfsetstrokecolor{dialinecolor}
\node[anchor=west] at (51.00\du,14.50\du){$i_2$};
% setfont left to latex
\definecolor{dialinecolor}{rgb}{0.000000, 0.000000, 0.000000}
\pgfsetstrokecolor{dialinecolor}
\node[anchor=west] at (56.00\du,14.50\du){$j_2$};
\definecolor{dialinecolor}{rgb}{1.000000, 1.000000, 1.000000}
\pgfsetfillcolor{dialinecolor}
\pgfpathellipse{\pgfpoint{66.463800\du}{6.570000\du}}{\pgfpoint{3.550000\du}{0\du}}{\pgfpoint{0\du}{3.450000\du}}
\pgfusepath{fill}
\pgfsetlinewidth{0.100000\du}
\pgfsetdash{}{0pt}
\pgfsetdash{}{0pt}
\definecolor{dialinecolor}{rgb}{0.000000, 0.000000, 0.000000}
\pgfsetstrokecolor{dialinecolor}
\pgfpathellipse{\pgfpoint{66.463800\du}{6.570000\du}}{\pgfpoint{3.550000\du}{0\du}}{\pgfpoint{0\du}{3.450000\du}}
\pgfusepath{stroke}
\pgfsetlinewidth{0.100000\du}
\pgfsetdash{}{0pt}
\pgfsetdash{}{0pt}
\pgfsetbuttcap
{
\definecolor{dialinecolor}{rgb}{0.000000, 0.000000, 0.000000}
\pgfsetfillcolor{dialinecolor}
% was here!!!
\definecolor{dialinecolor}{rgb}{0.000000, 0.000000, 0.000000}
\pgfsetstrokecolor{dialinecolor}
\draw (70.013800\du,6.570000\du)--(74.387000\du,6.570000\du);
}
\pgfsetlinewidth{0.100000\du}
\pgfsetdash{}{0pt}
\pgfsetdash{}{0pt}
\pgfsetbuttcap
{
\definecolor{dialinecolor}{rgb}{0.000000, 0.000000, 0.000000}
\pgfsetfillcolor{dialinecolor}
% was here!!!
\definecolor{dialinecolor}{rgb}{0.000000, 0.000000, 0.000000}
\pgfsetstrokecolor{dialinecolor}
\draw (58.613200\du,6.570000\du)--(62.913800\du,6.570000\du);
}
\pgfsetlinewidth{0.100000\du}
\pgfsetdash{}{0pt}
\pgfsetdash{}{0pt}
\pgfsetbuttcap
{
\definecolor{dialinecolor}{rgb}{0.000000, 0.000000, 0.000000}
\pgfsetfillcolor{dialinecolor}
% was here!!!
\definecolor{dialinecolor}{rgb}{0.000000, 0.000000, 0.000000}
\pgfsetstrokecolor{dialinecolor}
\draw (63.953600\du,9.009500\du)--(63.953600\du,13.00\du);
}
\pgfsetlinewidth{0.100000\du}
\pgfsetdash{}{0pt}
\pgfsetdash{}{0pt}
\pgfsetbuttcap
{
\definecolor{dialinecolor}{rgb}{0.000000, 0.000000, 0.000000}
\pgfsetfillcolor{dialinecolor}
% was here!!!
\definecolor{dialinecolor}{rgb}{0.000000, 0.000000, 0.000000}
\pgfsetstrokecolor{dialinecolor}
\draw (68.974000\du,9.009500\du)--(68.974000\du,13.00\du);
}
% setfont left to latex
\definecolor{dialinecolor}{rgb}{0.000000, 0.000000, 0.000000}
\pgfsetstrokecolor{dialinecolor}
\node[anchor=west] at (62.30\du,14.50\du){$i_3$};
% setfont left to latex
\definecolor{dialinecolor}{rgb}{0.000000, 0.000000, 0.000000}
\pgfsetstrokecolor{dialinecolor}
\node[anchor=west] at (67.45\du,14.50\du){$j_3$};
\definecolor{dialinecolor}{rgb}{1.000000, 1.000000, 1.000000}
\pgfsetfillcolor{dialinecolor}
\pgfpathellipse{\pgfpoint{77.937000\du}{6.570000\du}}{\pgfpoint{3.550000\du}{0\du}}{\pgfpoint{0\du}{3.450000\du}}
\pgfusepath{fill}
\pgfsetlinewidth{0.100000\du}
\pgfsetdash{}{0pt}
\pgfsetdash{}{0pt}
\definecolor{dialinecolor}{rgb}{0.000000, 0.000000, 0.000000}
\pgfsetstrokecolor{dialinecolor}
\pgfpathellipse{\pgfpoint{77.937000\du}{6.570000\du}}{\pgfpoint{3.550000\du}{0\du}}{\pgfpoint{0\du}{3.450000\du}}
\pgfusepath{stroke}
\pgfsetlinewidth{0.100000\du}
\pgfsetdash{}{0pt}
\pgfsetdash{}{0pt}
\pgfsetbuttcap
{
\definecolor{dialinecolor}{rgb}{0.000000, 0.000000, 0.000000}
\pgfsetfillcolor{dialinecolor}
% was here!!!
\definecolor{dialinecolor}{rgb}{0.000000, 0.000000, 0.000000}
\pgfsetstrokecolor{dialinecolor}
\draw (75.426700\du,9.009500\du)--(75.426700\du,13.00\du);
}
\pgfsetlinewidth{0.100000\du}
\pgfsetdash{}{0pt}
\pgfsetdash{}{0pt}
\pgfsetbuttcap
{
\definecolor{dialinecolor}{rgb}{0.000000, 0.000000, 0.000000}
\pgfsetfillcolor{dialinecolor}
% was here!!!
\definecolor{dialinecolor}{rgb}{0.000000, 0.000000, 0.000000}
\pgfsetstrokecolor{dialinecolor}
\draw (80.447200\du,9.009500\du)--(80.447200\du,13.00\du);
}
% setfont left to latex
\definecolor{dialinecolor}{rgb}{0.000000, 0.000000, 0.000000}
\pgfsetstrokecolor{dialinecolor}
\node[anchor=west] at (74.00\du,14.50\du){$i_4$};
% setfont left to latex
\definecolor{dialinecolor}{rgb}{0.000000, 0.000000, 0.000000}
\pgfsetstrokecolor{dialinecolor}
\node[anchor=west] at (78.90\du,14.50\du){$j_4$};
\pgfsetlinewidth{0.100000\du}
\pgfsetdash{}{0pt}
\pgfsetdash{}{0pt}
\pgfsetbuttcap
{
\definecolor{dialinecolor}{rgb}{0.000000, 0.000000, 0.000000}
\pgfsetfillcolor{dialinecolor}
% was here!!!
\definecolor{dialinecolor}{rgb}{0.000000, 0.000000, 0.000000}
\pgfsetstrokecolor{dialinecolor}
\draw (81.487000\du,6.570000\du)--(85.903200\du,6.570000\du);
}
% setfont left to latex
\definecolor{dialinecolor}{rgb}{0.000000, 0.000000, 0.000000}
\pgfsetstrokecolor{dialinecolor}
\node[anchor=west] at (40.70\du,6.05\du){$\ten{A}^{(1)}$};
% setfont left to latex
\definecolor{dialinecolor}{rgb}{0.000000, 0.000000, 0.000000}
\pgfsetstrokecolor{dialinecolor}
\node[anchor=west] at (52.20\du,6.05\du){$\ten{A}^{(2)}$};
% setfont left to latex
\definecolor{dialinecolor}{rgb}{0.000000, 0.000000, 0.000000}
\pgfsetstrokecolor{dialinecolor}
\node[anchor=west] at (63.70\du,6.05\du){$\ten{A}^{(3)}$};
% setfont left to latex
\definecolor{dialinecolor}{rgb}{0.000000, 0.000000, 0.000000}
\pgfsetstrokecolor{dialinecolor}
\node[anchor=west] at (75.20\du,6.05\du){$\ten{A}^{(4)}$};
% setfont left to latex
\definecolor{dialinecolor}{rgb}{0.000000, 0.000000, 0.000000}
\pgfsetstrokecolor{dialinecolor}
\node[anchor=west] at (36.00\du,5.60\du){$r_1$};
% setfont left to latex
\definecolor{dialinecolor}{rgb}{0.000000, 0.000000, 0.000000}
\pgfsetstrokecolor{dialinecolor}
\node[anchor=west] at (47.50\du,5.60\du){$r_2$};
% setfont left to latex
\definecolor{dialinecolor}{rgb}{0.000000, 0.000000, 0.000000}
\pgfsetstrokecolor{dialinecolor}
\node[anchor=west] at (59.00\du,5.60\du){$r_3$};
% setfont left to latex
\definecolor{dialinecolor}{rgb}{0.000000, 0.000000, 0.000000}
\pgfsetstrokecolor{dialinecolor}
\node[anchor=west] at (70.20\du,5.60\du){$r_4$};
% setfont left to latex
\definecolor{dialinecolor}{rgb}{0.000000, 0.000000, 0.000000}
\pgfsetstrokecolor{dialinecolor}
\node[anchor=west] at (81.50\du,5.60\du){$r_5$};
\end{tikzpicture}

%% file: figures/UTN.tex
% Graphic for TeX using PGF
% Title: /home/kim/Dropbox/Work/Papers/Journal/19.TNMoesp/figures/UTN.dia
% Creator: Dia v0.97.2
% CreationDate: Mon Sep 11 08:30:44 2017
% For: kim
% \usepackage{tikz}
% The following commands are not supported in PSTricks at present
% We define them conditionally, so when they are implemented,
% this pgf file will use them.
\ifx\du\undefined
  \newlength{\du}
\fi
\setlength{\du}{4\unitlength}
\begin{tikzpicture}
\pgftransformxscale{1.000000}
\pgftransformyscale{-1.000000}
\definecolor{dialinecolor}{rgb}{0.000000, 0.000000, 0.000000}
\pgfsetstrokecolor{dialinecolor}
\definecolor{dialinecolor}{rgb}{1.000000, 1.000000, 1.000000}
\pgfsetfillcolor{dialinecolor}
\definecolor{dialinecolor}{rgb}{1.000000, 1.000000, 1.000000}
\pgfsetfillcolor{dialinecolor}
\pgfpathellipse{\pgfpoint{13.090120\du}{11.015600\du}}{\pgfpoint{3.550000\du}{0\du}}{\pgfpoint{0\du}{3.450000\du}}
\pgfusepath{fill}
\pgfsetlinewidth{0.100000\du}
\pgfsetdash{}{0pt}
\pgfsetdash{}{0pt}
\definecolor{dialinecolor}{rgb}{0.000000, 0.000000, 0.000000}
\pgfsetstrokecolor{dialinecolor}
\pgfpathellipse{\pgfpoint{13.090120\du}{11.015600\du}}{\pgfpoint{3.550000\du}{0\du}}{\pgfpoint{0\du}{3.450000\du}}
\pgfusepath{stroke}
\pgfsetlinewidth{0.100000\du}
\pgfsetdash{}{0pt}
\pgfsetdash{}{0pt}
\pgfsetbuttcap
{
\definecolor{dialinecolor}{rgb}{0.000000, 0.000000, 0.000000}
\pgfsetfillcolor{dialinecolor}
% was here!!!
\definecolor{dialinecolor}{rgb}{0.000000, 0.000000, 0.000000}
\pgfsetstrokecolor{dialinecolor}
\draw (16.640140\du,11.015600\du)--(21.013340\du,11.015600\du);
}
\pgfsetlinewidth{0.100000\du}
\pgfsetdash{}{0pt}
\pgfsetdash{}{0pt}
\pgfsetbuttcap
{
\definecolor{dialinecolor}{rgb}{0.000000, 0.000000, 0.000000}
\pgfsetfillcolor{dialinecolor}
% was here!!!
\definecolor{dialinecolor}{rgb}{0.000000, 0.000000, 0.000000}
\pgfsetstrokecolor{dialinecolor}
\draw (5.395120\du,11.005600\du)--(9.540120\du,11.015600\du);
}
\pgfsetlinewidth{0.100000\du}
\pgfsetdash{}{0pt}
\pgfsetdash{}{0pt}
\pgfsetbuttcap
{
\definecolor{dialinecolor}{rgb}{0.000000, 0.000000, 0.000000}
\pgfsetfillcolor{dialinecolor}
% was here!!!
\definecolor{dialinecolor}{rgb}{0.000000, 0.000000, 0.000000}
\pgfsetstrokecolor{dialinecolor}
\draw (13.090140\du,14.465600\du)--(13.110340\du,18.50\du);
}
\pgfsetlinewidth{0.100000\du}
\pgfsetdash{}{0pt}
\pgfsetdash{}{0pt}
\pgfsetbuttcap
{
\definecolor{dialinecolor}{rgb}{0.000000, 0.000000, 0.000000}
\pgfsetfillcolor{dialinecolor}
% was here!!!
\definecolor{dialinecolor}{rgb}{0.000000, 0.000000, 0.000000}
\pgfsetstrokecolor{dialinecolor}
\draw (13.090140\du,3.140500\du)--(13.090140\du,7.565600\du);
}
% setfont left to latex
\definecolor{dialinecolor}{rgb}{0.000000, 0.000000, 0.000000}
\pgfsetstrokecolor{dialinecolor}
\node[anchor=west] at (12.613340\du,17.00\du){$m$};
\definecolor{dialinecolor}{rgb}{1.000000, 1.000000, 1.000000}
\pgfsetfillcolor{dialinecolor}
\pgfpathellipse{\pgfpoint{24.563340\du}{11.015600\du}}{\pgfpoint{3.550000\du}{0\du}}{\pgfpoint{0\du}{3.450000\du}}
\pgfusepath{fill}
\pgfsetlinewidth{0.100000\du}
\pgfsetdash{}{0pt}
\pgfsetdash{}{0pt}
\definecolor{dialinecolor}{rgb}{0.000000, 0.000000, 0.000000}
\pgfsetstrokecolor{dialinecolor}
\pgfpathellipse{\pgfpoint{24.563340\du}{11.015600\du}}{\pgfpoint{3.550000\du}{0\du}}{\pgfpoint{0\du}{3.450000\du}}
\pgfusepath{stroke}
\pgfsetlinewidth{0.100000\du}
\pgfsetdash{}{0pt}
\pgfsetdash{}{0pt}
\pgfsetbuttcap
{
\definecolor{dialinecolor}{rgb}{0.000000, 0.000000, 0.000000}
\pgfsetfillcolor{dialinecolor}
% was here!!!
\definecolor{dialinecolor}{rgb}{0.000000, 0.000000, 0.000000}
\pgfsetstrokecolor{dialinecolor}
\draw (24.563340\du,14.465600\du)--(24.534040\du,18.50\du);
}
\pgfsetlinewidth{0.100000\du}
\pgfsetdash{}{0pt}
\pgfsetdash{}{0pt}
\pgfsetbuttcap
{
\definecolor{dialinecolor}{rgb}{0.000000, 0.000000, 0.000000}
\pgfsetfillcolor{dialinecolor}
% was here!!!
\definecolor{dialinecolor}{rgb}{0.000000, 0.000000, 0.000000}
\pgfsetstrokecolor{dialinecolor}
\draw (24.563340\du,3.060000\du)--(24.563340\du,7.565600\du);
}
\pgfsetlinewidth{0.100000\du}
\pgfsetdash{}{0pt}
\pgfsetdash{}{0pt}
\pgfsetbuttcap
{
\definecolor{dialinecolor}{rgb}{0.000000, 0.000000, 0.000000}
\pgfsetfillcolor{dialinecolor}
% was here!!!
\definecolor{dialinecolor}{rgb}{0.000000, 0.000000, 0.000000}
\pgfsetstrokecolor{dialinecolor}
\draw (38.137440\du,11.169800\du)--(42.510640\du,11.169800\du);
}
\pgfsetlinewidth{0.100000\du}
\pgfsetdash{}{0pt}
\pgfsetdash{}{0pt}
\pgfsetbuttcap
{
\definecolor{dialinecolor}{rgb}{0.000000, 0.000000, 0.000000}
\pgfsetfillcolor{dialinecolor}
% was here!!!
\definecolor{dialinecolor}{rgb}{0.000000, 0.000000, 0.000000}
\pgfsetstrokecolor{dialinecolor}
\draw (28.113340\du,11.015600\du)--(32.413940\du,11.015600\du);
}
\definecolor{dialinecolor}{rgb}{1.000000, 1.000000, 1.000000}
\pgfsetfillcolor{dialinecolor}
\pgfpathellipse{\pgfpoint{46.060640\du}{11.169800\du}}{\pgfpoint{3.550000\du}{0\du}}{\pgfpoint{0\du}{3.450000\du}}
\pgfusepath{fill}
\pgfsetlinewidth{0.100000\du}
\pgfsetdash{}{0pt}
\pgfsetdash{}{0pt}
\definecolor{dialinecolor}{rgb}{0.000000, 0.000000, 0.000000}
\pgfsetstrokecolor{dialinecolor}
\pgfpathellipse{\pgfpoint{46.060640\du}{11.169800\du}}{\pgfpoint{3.550000\du}{0\du}}{\pgfpoint{0\du}{3.450000\du}}
\pgfusepath{stroke}
\pgfsetlinewidth{0.100000\du}
\pgfsetdash{}{0pt}
\pgfsetdash{}{0pt}
\pgfsetbuttcap
{
\definecolor{dialinecolor}{rgb}{0.000000, 0.000000, 0.000000}
\pgfsetfillcolor{dialinecolor}
% was here!!!
\definecolor{dialinecolor}{rgb}{0.000000, 0.000000, 0.000000}
\pgfsetstrokecolor{dialinecolor}
\draw (46.060640\du,14.619800\du)--(46.077140\du,18.50\du);
}
\pgfsetlinewidth{0.100000\du}
\pgfsetdash{}{0pt}
\pgfsetdash{}{0pt}
\pgfsetbuttcap
{
\definecolor{dialinecolor}{rgb}{0.000000, 0.000000, 0.000000}
\pgfsetfillcolor{dialinecolor}
% was here!!!
\definecolor{dialinecolor}{rgb}{0.000000, 0.000000, 0.000000}
\pgfsetstrokecolor{dialinecolor}
\draw (46.013340\du,3.403900\du)--(46.060640\du,7.719800\du);
}
\pgfsetlinewidth{0.100000\du}
\pgfsetdash{}{0pt}
\pgfsetdash{}{0pt}
\pgfsetbuttcap
{
\definecolor{dialinecolor}{rgb}{0.000000, 0.000000, 0.000000}
\pgfsetfillcolor{dialinecolor}
% was here!!!
\definecolor{dialinecolor}{rgb}{0.000000, 0.000000, 0.000000}
\pgfsetstrokecolor{dialinecolor}
\draw (49.610640\du,11.169800\du)--(54.026840\du,11.169800\du);
}
% setfont left to latex
\definecolor{dialinecolor}{rgb}{0.000000, 0.000000, 0.000000}
\pgfsetstrokecolor{dialinecolor}
\node[anchor=west] at (10.60\du,11.00\du){$\ten{U}^{(1)}$};
% setfont left to latex
\definecolor{dialinecolor}{rgb}{0.000000, 0.000000, 0.000000}
\pgfsetstrokecolor{dialinecolor}
\node[anchor=west] at (22.00\du,11.00\du){$\ten{U}^{(2)}$};
% setfont left to latex
\definecolor{dialinecolor}{rgb}{0.000000, 0.000000, 0.000000}
\pgfsetstrokecolor{dialinecolor}
\node[anchor=west] at (45.50\du,5.814300\du){$kN$};
% setfont left to latex
\definecolor{dialinecolor}{rgb}{0.000000, 0.000000, 0.000000}
\pgfsetstrokecolor{dialinecolor}
\node[anchor=west] at (23.948140\du,17.00\du){$m$};
% setfont left to latex
\definecolor{dialinecolor}{rgb}{0.000000, 0.000000, 0.000000}
\pgfsetstrokecolor{dialinecolor}
\node[anchor=west] at (45.50\du,17.00\du){$m$};
% setfont left to latex
\definecolor{dialinecolor}{rgb}{0.000000, 0.000000, 0.000000}
\pgfsetstrokecolor{dialinecolor}
\node[anchor=west] at (12.613340\du,5.312000\du){$1$};
% setfont left to latex
\definecolor{dialinecolor}{rgb}{0.000000, 0.000000, 0.000000}
\pgfsetstrokecolor{dialinecolor}
\node[anchor=west] at (23.948140\du,5.363800\du){$1$};
% setfont left to latex
\definecolor{dialinecolor}{rgb}{0.000000, 0.000000, 0.000000}
\pgfsetstrokecolor{dialinecolor}
\node[anchor=west] at (43.50\du,11.00\du){$\ten{U}^{(d)}$};
% setfont left to latex
\definecolor{dialinecolor}{rgb}{0.000000, 0.000000, 0.000000}
\pgfsetstrokecolor{dialinecolor}
\node[anchor=west] at (6.50\du,9.80\du){$1$};
% setfont left to latex
\definecolor{dialinecolor}{rgb}{0.000000, 0.000000, 0.000000}
\pgfsetstrokecolor{dialinecolor}
\node[anchor=west] at (50.00\du,9.80\du){$1$};
% setfont left to latex
\definecolor{dialinecolor}{rgb}{0.000000, 0.000000, 0.000000}
\pgfsetstrokecolor{dialinecolor}
\node[anchor=west] at (17.25\du,9.80\du){$r_2$};
% setfont left to latex
\definecolor{dialinecolor}{rgb}{0.000000, 0.000000, 0.000000}
\pgfsetstrokecolor{dialinecolor}
\node[anchor=west] at (28.50\du,9.80\du){$r_3$};
% setfont left to latex
\definecolor{dialinecolor}{rgb}{0.000000, 0.000000, 0.000000}
\pgfsetstrokecolor{dialinecolor}
\node[anchor=west] at (38.75\du,9.80\du){$r_d$};
\pgfsetlinewidth{0.100000\du}
\pgfsetdash{}{0pt}
\pgfsetdash{}{0pt}
\pgfsetbuttcap
\pgfsetmiterjoin
\pgfsetlinewidth{0.100000\du}
\pgfsetbuttcap
\pgfsetmiterjoin
\pgfsetdash{}{0pt}
\definecolor{dialinecolor}{rgb}{0.000000, 0.000000, 0.000000}
\pgfsetfillcolor{dialinecolor}
\pgfpathellipse{\pgfpoint{33.951940\du}{11.075000\du}}{\pgfpoint{0.425000\du}{0\du}}{\pgfpoint{0\du}{0.425000\du}}
\pgfusepath{fill}
\definecolor{dialinecolor}{rgb}{0.000000, 0.000000, 0.000000}
\pgfsetstrokecolor{dialinecolor}
\pgfpathellipse{\pgfpoint{33.951940\du}{11.075000\du}}{\pgfpoint{0.425000\du}{0\du}}{\pgfpoint{0\du}{0.425000\du}}
\pgfusepath{stroke}
\pgfsetbuttcap
\pgfsetmiterjoin
\pgfsetdash{}{0pt}
\definecolor{dialinecolor}{rgb}{0.000000, 0.000000, 0.000000}
\pgfsetstrokecolor{dialinecolor}
\pgfpathellipse{\pgfpoint{33.951940\du}{11.075000\du}}{\pgfpoint{0.425000\du}{0\du}}{\pgfpoint{0\du}{0.425000\du}}
\pgfusepath{stroke}
\pgfsetlinewidth{0.100000\du}
\pgfsetdash{}{0pt}
\pgfsetdash{}{0pt}
\pgfsetbuttcap
\pgfsetmiterjoin
\pgfsetlinewidth{0.100000\du}
\pgfsetbuttcap
\pgfsetmiterjoin
\pgfsetdash{}{0pt}
\definecolor{dialinecolor}{rgb}{0.000000, 0.000000, 0.000000}
\pgfsetfillcolor{dialinecolor}
\pgfpathellipse{\pgfpoint{35.591940\du}{11.075000\du}}{\pgfpoint{0.425000\du}{0\du}}{\pgfpoint{0\du}{0.425000\du}}
\pgfusepath{fill}
\definecolor{dialinecolor}{rgb}{0.000000, 0.000000, 0.000000}
\pgfsetstrokecolor{dialinecolor}
\pgfpathellipse{\pgfpoint{35.591940\du}{11.075000\du}}{\pgfpoint{0.425000\du}{0\du}}{\pgfpoint{0\du}{0.425000\du}}
\pgfusepath{stroke}
\pgfsetbuttcap
\pgfsetmiterjoin
\pgfsetdash{}{0pt}
\definecolor{dialinecolor}{rgb}{0.000000, 0.000000, 0.000000}
\pgfsetstrokecolor{dialinecolor}
\pgfpathellipse{\pgfpoint{35.591940\du}{11.075000\du}}{\pgfpoint{0.425000\du}{0\du}}{\pgfpoint{0\du}{0.425000\du}}
\pgfusepath{stroke}
\pgfsetlinewidth{0.100000\du}
\pgfsetdash{}{0pt}
\pgfsetdash{}{0pt}
\pgfsetbuttcap
\pgfsetmiterjoin
\pgfsetlinewidth{0.100000\du}
\pgfsetbuttcap
\pgfsetmiterjoin
\pgfsetdash{}{0pt}
\definecolor{dialinecolor}{rgb}{0.000000, 0.000000, 0.000000}
\pgfsetfillcolor{dialinecolor}
\pgfpathellipse{\pgfpoint{37.146940\du}{11.075000\du}}{\pgfpoint{0.425000\du}{0\du}}{\pgfpoint{0\du}{0.425000\du}}
\pgfusepath{fill}
\definecolor{dialinecolor}{rgb}{0.000000, 0.000000, 0.000000}
\pgfsetstrokecolor{dialinecolor}
\pgfpathellipse{\pgfpoint{37.146940\du}{11.075000\du}}{\pgfpoint{0.425000\du}{0\du}}{\pgfpoint{0\du}{0.425000\du}}
\pgfusepath{stroke}
\pgfsetbuttcap
\pgfsetmiterjoin
\pgfsetdash{}{0pt}
\definecolor{dialinecolor}{rgb}{0.000000, 0.000000, 0.000000}
\pgfsetstrokecolor{dialinecolor}
\pgfpathellipse{\pgfpoint{37.146940\du}{11.075000\du}}{\pgfpoint{0.425000\du}{0\du}}{\pgfpoint{0\du}{0.425000\du}}
\pgfusepath{stroke}
\end{tikzpicture}

%% file: figures/TNSVD.tex
% Graphic for TeX using PGF
% Title: /home/kim/Dropbox/Work/Papers/Journal/19.TNMoesp/figures/TNSVD.dia
% Creator: Dia v0.97.2
% CreationDate: Mon Sep 11 09:06:42 2017
% For: kim
% \usepackage{tikz}
% The following commands are not supported in PSTricks at present
% We define them conditionally, so when they are implemented,
% this pgf file will use them.
\ifx\du\undefined
  \newlength{\du}
\fi
\setlength{\du}{4\unitlength}
\begin{tikzpicture}
\pgftransformxscale{1.000000}
\pgftransformyscale{-1.000000}
\definecolor{dialinecolor}{rgb}{0.000000, 0.000000, 0.000000}
\pgfsetstrokecolor{dialinecolor}
\definecolor{dialinecolor}{rgb}{1.000000, 1.000000, 1.000000}
\pgfsetfillcolor{dialinecolor}
\definecolor{dialinecolor}{rgb}{1.000000, 1.000000, 1.000000}
\pgfsetfillcolor{dialinecolor}
\pgfpathellipse{\pgfpoint{11.316580\du}{22.195300\du}}{\pgfpoint{3.550000\du}{0\du}}{\pgfpoint{0\du}{3.450000\du}}
\pgfusepath{fill}
\pgfsetlinewidth{0.100000\du}
\pgfsetdash{}{0pt}
\pgfsetdash{}{0pt}
\definecolor{dialinecolor}{rgb}{0.000000, 0.000000, 0.000000}
\pgfsetstrokecolor{dialinecolor}
\pgfpathellipse{\pgfpoint{11.316580\du}{22.195300\du}}{\pgfpoint{3.550000\du}{0\du}}{\pgfpoint{0\du}{3.450000\du}}
\pgfusepath{stroke}
\pgfsetlinewidth{0.100000\du}
\pgfsetdash{}{0pt}
\pgfsetdash{}{0pt}
\pgfsetbuttcap
{
\definecolor{dialinecolor}{rgb}{0.000000, 0.000000, 0.000000}
\pgfsetfillcolor{dialinecolor}
% was here!!!
\definecolor{dialinecolor}{rgb}{0.000000, 0.000000, 0.000000}
\pgfsetstrokecolor{dialinecolor}
\draw (14.866600\du,22.195300\du)--(19.239800\du,22.195300\du);
}
\pgfsetlinewidth{0.100000\du}
\pgfsetdash{}{0pt}
\pgfsetdash{}{0pt}
\pgfsetbuttcap
{
\definecolor{dialinecolor}{rgb}{0.000000, 0.000000, 0.000000}
\pgfsetfillcolor{dialinecolor}
% was here!!!
\definecolor{dialinecolor}{rgb}{0.000000, 0.000000, 0.000000}
\pgfsetstrokecolor{dialinecolor}
\draw (3.621580\du,22.185300\du)--(7.766580\du,22.195300\du);
}
\pgfsetlinewidth{0.100000\du}
\pgfsetdash{}{0pt}
\pgfsetdash{}{0pt}
\pgfsetbuttcap
{
\definecolor{dialinecolor}{rgb}{0.000000, 0.000000, 0.000000}
\pgfsetfillcolor{dialinecolor}
% was here!!!
\definecolor{dialinecolor}{rgb}{0.000000, 0.000000, 0.000000}
\pgfsetstrokecolor{dialinecolor}
\draw (11.316600\du,25.645300\du)--(11.336800\du,30.224900\du);
}
\pgfsetlinewidth{0.100000\du}
\pgfsetdash{}{0pt}
\pgfsetdash{}{0pt}
\pgfsetbuttcap
{
\definecolor{dialinecolor}{rgb}{0.000000, 0.000000, 0.000000}
\pgfsetfillcolor{dialinecolor}
% was here!!!
\definecolor{dialinecolor}{rgb}{0.000000, 0.000000, 0.000000}
\pgfsetstrokecolor{dialinecolor}
\draw (11.316600\du,14.320200\du)--(11.316600\du,18.745300\du);
}
% setfont left to latex
\definecolor{dialinecolor}{rgb}{0.000000, 0.000000, 0.000000}
\pgfsetstrokecolor{dialinecolor}
\node[anchor=west] at (10.839800\du,28.50\du){$m$};
\definecolor{dialinecolor}{rgb}{1.000000, 1.000000, 1.000000}
\pgfsetfillcolor{dialinecolor}
\pgfpathellipse{\pgfpoint{22.789800\du}{22.195300\du}}{\pgfpoint{3.550000\du}{0\du}}{\pgfpoint{0\du}{3.450000\du}}
\pgfusepath{fill}
\pgfsetlinewidth{0.100000\du}
\pgfsetdash{}{0pt}
\pgfsetdash{}{0pt}
\definecolor{dialinecolor}{rgb}{0.000000, 0.000000, 0.000000}
\pgfsetstrokecolor{dialinecolor}
\pgfpathellipse{\pgfpoint{22.789800\du}{22.195300\du}}{\pgfpoint{3.550000\du}{0\du}}{\pgfpoint{0\du}{3.450000\du}}
\pgfusepath{stroke}
\pgfsetlinewidth{0.100000\du}
\pgfsetdash{}{0pt}
\pgfsetdash{}{0pt}
\pgfsetbuttcap
{
\definecolor{dialinecolor}{rgb}{0.000000, 0.000000, 0.000000}
\pgfsetfillcolor{dialinecolor}
% was here!!!
\definecolor{dialinecolor}{rgb}{0.000000, 0.000000, 0.000000}
\pgfsetstrokecolor{dialinecolor}
\draw (22.789800\du,25.645300\du)--(22.760500\du,30.165400\du);
}
\pgfsetlinewidth{0.100000\du}
\pgfsetdash{}{0pt}
\pgfsetdash{}{0pt}
\pgfsetbuttcap
{
\definecolor{dialinecolor}{rgb}{0.000000, 0.000000, 0.000000}
\pgfsetfillcolor{dialinecolor}
% was here!!!
\definecolor{dialinecolor}{rgb}{0.000000, 0.000000, 0.000000}
\pgfsetstrokecolor{dialinecolor}
\draw (22.789800\du,14.239700\du)--(22.789800\du,18.745300\du);
}
\pgfsetlinewidth{0.100000\du}
\pgfsetdash{}{0pt}
\pgfsetdash{}{0pt}
\pgfsetbuttcap
{
\definecolor{dialinecolor}{rgb}{0.000000, 0.000000, 0.000000}
\pgfsetfillcolor{dialinecolor}
% was here!!!
\definecolor{dialinecolor}{rgb}{0.000000, 0.000000, 0.000000}
\pgfsetstrokecolor{dialinecolor}
\draw (36.713899\du,22.249500\du)--(41.087099\du,22.249500\du);
}
\pgfsetlinewidth{0.100000\du}
\pgfsetdash{}{0pt}
\pgfsetdash{}{0pt}
\pgfsetbuttcap
{
\definecolor{dialinecolor}{rgb}{0.000000, 0.000000, 0.000000}
\pgfsetfillcolor{dialinecolor}
% was here!!!
\definecolor{dialinecolor}{rgb}{0.000000, 0.000000, 0.000000}
\pgfsetstrokecolor{dialinecolor}
\draw (26.339800\du,22.195300\du)--(30.640400\du,22.195300\du);
}
\definecolor{dialinecolor}{rgb}{1.000000, 1.000000, 1.000000}
\pgfsetfillcolor{dialinecolor}
\pgfpathellipse{\pgfpoint{44.637099\du}{22.249500\du}}{\pgfpoint{3.550000\du}{0\du}}{\pgfpoint{0\du}{3.450000\du}}
\pgfusepath{fill}
\pgfsetlinewidth{0.100000\du}
\pgfsetdash{}{0pt}
\pgfsetdash{}{0pt}
\definecolor{dialinecolor}{rgb}{0.000000, 0.000000, 0.000000}
\pgfsetstrokecolor{dialinecolor}
\pgfpathellipse{\pgfpoint{44.637099\du}{22.249500\du}}{\pgfpoint{3.550000\du}{0\du}}{\pgfpoint{0\du}{3.450000\du}}
\pgfusepath{stroke}
\pgfsetlinewidth{0.100000\du}
\pgfsetdash{}{0pt}
\pgfsetdash{}{0pt}
\pgfsetbuttcap
{
\definecolor{dialinecolor}{rgb}{0.000000, 0.000000, 0.000000}
\pgfsetfillcolor{dialinecolor}
% was here!!!
\definecolor{dialinecolor}{rgb}{0.000000, 0.000000, 0.000000}
\pgfsetstrokecolor{dialinecolor}
\draw (44.637099\du,25.699500\du)--(44.653599\du,30.214000\du);
}
\pgfsetlinewidth{0.100000\du}
\pgfsetdash{}{0pt}
\pgfsetdash{}{0pt}
\pgfsetbuttcap
{
\definecolor{dialinecolor}{rgb}{0.000000, 0.000000, 0.000000}
\pgfsetfillcolor{dialinecolor}
% was here!!!
\definecolor{dialinecolor}{rgb}{0.000000, 0.000000, 0.000000}
\pgfsetstrokecolor{dialinecolor}
\draw (44.589799\du,14.483570\du)--(44.637099\du,18.799500\du);
}
\pgfsetlinewidth{0.100000\du}
\pgfsetdash{}{0pt}
\pgfsetdash{}{0pt}
\pgfsetbuttcap
{
\definecolor{dialinecolor}{rgb}{0.000000, 0.000000, 0.000000}
\pgfsetfillcolor{dialinecolor}
% was here!!!
\definecolor{dialinecolor}{rgb}{0.000000, 0.000000, 0.000000}
\pgfsetstrokecolor{dialinecolor}
\draw (48.187099\du,22.249500\du)--(52.603299\du,22.249500\du);
}
% setfont left to latex
\definecolor{dialinecolor}{rgb}{0.000000, 0.000000, 0.000000}
\pgfsetstrokecolor{dialinecolor}
\node[anchor=west] at (8.00\du,22.10\du){$\ten{W}^{(1)}$};
% setfont left to latex
\definecolor{dialinecolor}{rgb}{0.000000, 0.000000, 0.000000}
\pgfsetstrokecolor{dialinecolor}
\node[anchor=west] at (19.80\du,22.10\du){$\ten{W}^{(2)}$};
% setfont left to latex
\definecolor{dialinecolor}{rgb}{0.000000, 0.000000, 0.000000}
\pgfsetstrokecolor{dialinecolor}
\node[anchor=west] at (41.25\du,22.10\du){$\ten{W}^{(d)}$};
% setfont left to latex
\definecolor{dialinecolor}{rgb}{0.000000, 0.000000, 0.000000}
\pgfsetstrokecolor{dialinecolor}
\node[anchor=west] at (44.00\du,16.894000\du){$N$};
\definecolor{dialinecolor}{rgb}{1.000000, 1.000000, 1.000000}
\pgfsetfillcolor{dialinecolor}
\pgfpathellipse{\pgfpoint{44.589799\du}{11.033570\du}}{\pgfpoint{3.550000\du}{0\du}}{\pgfpoint{0\du}{3.450000\du}}
\pgfusepath{fill}
\pgfsetlinewidth{0.100000\du}
\pgfsetdash{}{0pt}
\pgfsetdash{}{0pt}
\definecolor{dialinecolor}{rgb}{0.000000, 0.000000, 0.000000}
\pgfsetstrokecolor{dialinecolor}
\pgfpathellipse{\pgfpoint{44.589799\du}{11.033570\du}}{\pgfpoint{3.550000\du}{0\du}}{\pgfpoint{0\du}{3.450000\du}}
\pgfusepath{stroke}
\pgfsetlinewidth{0.100000\du}
\pgfsetdash{}{0pt}
\pgfsetdash{}{0pt}
\pgfsetbuttcap
{
\definecolor{dialinecolor}{rgb}{0.000000, 0.000000, 0.000000}
\pgfsetfillcolor{dialinecolor}
% was here!!!
\definecolor{dialinecolor}{rgb}{0.000000, 0.000000, 0.000000}
\pgfsetstrokecolor{dialinecolor}
\draw (44.594699\du,2.944120\du)--(44.589799\du,7.583570\du);
}
% setfont left to latex
\definecolor{dialinecolor}{rgb}{0.000000, 0.000000, 0.000000}
\pgfsetstrokecolor{dialinecolor}
\node[anchor=west] at (44.00\du,5.241270\du){$N$};
% setfont left to latex
\definecolor{dialinecolor}{rgb}{0.000000, 0.000000, 0.000000}
\pgfsetstrokecolor{dialinecolor}
\node[anchor=west] at (42.610899\du,11.25\du){$\mat{T}$};
\definecolor{dialinecolor}{rgb}{1.000000, 1.000000, 1.000000}
\pgfsetfillcolor{dialinecolor}
\pgfpathellipse{\pgfpoint{44.594699\du}{-0.505880\du}}{\pgfpoint{3.550000\du}{0\du}}{\pgfpoint{0\du}{3.450000\du}}
\pgfusepath{fill}
\pgfsetlinewidth{0.100000\du}
\pgfsetdash{}{0pt}
\pgfsetdash{}{0pt}
\definecolor{dialinecolor}{rgb}{0.000000, 0.000000, 0.000000}
\pgfsetstrokecolor{dialinecolor}
\pgfpathellipse{\pgfpoint{44.594699\du}{-0.505880\du}}{\pgfpoint{3.550000\du}{0\du}}{\pgfpoint{0\du}{3.450000\du}}
\pgfusepath{stroke}
% setfont left to latex
\definecolor{dialinecolor}{rgb}{0.000000, 0.000000, 0.000000}
\pgfsetstrokecolor{dialinecolor}
\node[anchor=west] at (42.610899\du,-0.30\du){$\mat{Q}$};
\pgfsetlinewidth{0.100000\du}
\pgfsetdash{}{0pt}
\pgfsetdash{}{0pt}
\pgfsetbuttcap
{
\definecolor{dialinecolor}{rgb}{0.000000, 0.000000, 0.000000}
\pgfsetfillcolor{dialinecolor}
% was here!!!
\definecolor{dialinecolor}{rgb}{0.000000, 0.000000, 0.000000}
\pgfsetstrokecolor{dialinecolor}
\draw (44.599599\du,-8.555830\du)--(44.594699\du,-3.955880\du);
}
% setfont left to latex
\definecolor{dialinecolor}{rgb}{0.000000, 0.000000, 0.000000}
\pgfsetstrokecolor{dialinecolor}
\node[anchor=west] at (22.174600\du,28.50\du){$m$};
% setfont left to latex
\definecolor{dialinecolor}{rgb}{0.000000, 0.000000, 0.000000}
\pgfsetstrokecolor{dialinecolor}
\node[anchor=west] at (44.00\du,28.50\du){$mk$};
% setfont left to latex
\definecolor{dialinecolor}{rgb}{0.000000, 0.000000, 0.000000}
\pgfsetstrokecolor{dialinecolor}
\node[anchor=west] at (44.00\du,-5.869990\du){$N$};
% setfont left to latex
\definecolor{dialinecolor}{rgb}{0.000000, 0.000000, 0.000000}
\pgfsetstrokecolor{dialinecolor}
\node[anchor=west] at (10.839800\du,16.491700\du){$1$};
% setfont left to latex
\definecolor{dialinecolor}{rgb}{0.000000, 0.000000, 0.000000}
\pgfsetstrokecolor{dialinecolor}
\node[anchor=west] at (22.174600\du,16.543500\du){$1$};
\pgfsetlinewidth{0.100000\du}
\pgfsetdash{}{0pt}
\pgfsetdash{}{0pt}
\pgfsetbuttcap
\pgfsetmiterjoin
\pgfsetlinewidth{0.100000\du}
\pgfsetbuttcap
\pgfsetmiterjoin
\pgfsetdash{}{0pt}
\definecolor{dialinecolor}{rgb}{0.000000, 0.000000, 0.000000}
\pgfsetfillcolor{dialinecolor}
\pgfpathellipse{\pgfpoint{32.020000\du}{22.285000\du}}{\pgfpoint{0.425000\du}{0\du}}{\pgfpoint{0\du}{0.425000\du}}
\pgfusepath{fill}
\definecolor{dialinecolor}{rgb}{0.000000, 0.000000, 0.000000}
\pgfsetstrokecolor{dialinecolor}
\pgfpathellipse{\pgfpoint{32.020000\du}{22.285000\du}}{\pgfpoint{0.425000\du}{0\du}}{\pgfpoint{0\du}{0.425000\du}}
\pgfusepath{stroke}
\pgfsetbuttcap
\pgfsetmiterjoin
\pgfsetdash{}{0pt}
\definecolor{dialinecolor}{rgb}{0.000000, 0.000000, 0.000000}
\pgfsetstrokecolor{dialinecolor}
\pgfpathellipse{\pgfpoint{32.020000\du}{22.285000\du}}{\pgfpoint{0.425000\du}{0\du}}{\pgfpoint{0\du}{0.425000\du}}
\pgfusepath{stroke}
\pgfsetlinewidth{0.100000\du}
\pgfsetdash{}{0pt}
\pgfsetdash{}{0pt}
\pgfsetbuttcap
\pgfsetmiterjoin
\pgfsetlinewidth{0.100000\du}
\pgfsetbuttcap
\pgfsetmiterjoin
\pgfsetdash{}{0pt}
\definecolor{dialinecolor}{rgb}{0.000000, 0.000000, 0.000000}
\pgfsetfillcolor{dialinecolor}
\pgfpathellipse{\pgfpoint{33.660000\du}{22.285000\du}}{\pgfpoint{0.425000\du}{0\du}}{\pgfpoint{0\du}{0.425000\du}}
\pgfusepath{fill}
\definecolor{dialinecolor}{rgb}{0.000000, 0.000000, 0.000000}
\pgfsetstrokecolor{dialinecolor}
\pgfpathellipse{\pgfpoint{33.660000\du}{22.285000\du}}{\pgfpoint{0.425000\du}{0\du}}{\pgfpoint{0\du}{0.425000\du}}
\pgfusepath{stroke}
\pgfsetbuttcap
\pgfsetmiterjoin
\pgfsetdash{}{0pt}
\definecolor{dialinecolor}{rgb}{0.000000, 0.000000, 0.000000}
\pgfsetstrokecolor{dialinecolor}
\pgfpathellipse{\pgfpoint{33.660000\du}{22.285000\du}}{\pgfpoint{0.425000\du}{0\du}}{\pgfpoint{0\du}{0.425000\du}}
\pgfusepath{stroke}
\pgfsetlinewidth{0.100000\du}
\pgfsetdash{}{0pt}
\pgfsetdash{}{0pt}
\pgfsetbuttcap
\pgfsetmiterjoin
\pgfsetlinewidth{0.100000\du}
\pgfsetbuttcap
\pgfsetmiterjoin
\pgfsetdash{}{0pt}
\definecolor{dialinecolor}{rgb}{0.000000, 0.000000, 0.000000}
\pgfsetfillcolor{dialinecolor}
\pgfpathellipse{\pgfpoint{35.215000\du}{22.285000\du}}{\pgfpoint{0.425000\du}{0\du}}{\pgfpoint{0\du}{0.425000\du}}
\pgfusepath{fill}
\definecolor{dialinecolor}{rgb}{0.000000, 0.000000, 0.000000}
\pgfsetstrokecolor{dialinecolor}
\pgfpathellipse{\pgfpoint{35.215000\du}{22.285000\du}}{\pgfpoint{0.425000\du}{0\du}}{\pgfpoint{0\du}{0.425000\du}}
\pgfusepath{stroke}
\pgfsetbuttcap
\pgfsetmiterjoin
\pgfsetdash{}{0pt}
\definecolor{dialinecolor}{rgb}{0.000000, 0.000000, 0.000000}
\pgfsetstrokecolor{dialinecolor}
\pgfpathellipse{\pgfpoint{35.215000\du}{22.285000\du}}{\pgfpoint{0.425000\du}{0\du}}{\pgfpoint{0\du}{0.425000\du}}
\pgfusepath{stroke}
% setfont left to latex
\definecolor{dialinecolor}{rgb}{0.000000, 0.000000, 0.000000}
\pgfsetstrokecolor{dialinecolor}
\node[anchor=west] at (5.00\du,21.20\du){$1$};
% setfont left to latex
\definecolor{dialinecolor}{rgb}{0.000000, 0.000000, 0.000000}
\pgfsetstrokecolor{dialinecolor}
\node[anchor=west] at (15.40\du,21.20\du){$r_2$};
% setfont left to latex
\definecolor{dialinecolor}{rgb}{0.000000, 0.000000, 0.000000}
\pgfsetstrokecolor{dialinecolor}
\node[anchor=west] at (26.50\du,21.20\du){$r_3$};
% setfont left to latex
\definecolor{dialinecolor}{rgb}{0.000000, 0.000000, 0.000000}
\pgfsetstrokecolor{dialinecolor}
\node[anchor=west] at (37.50\du,21.20\du){$r_d$};
% setfont left to latex
\definecolor{dialinecolor}{rgb}{0.000000, 0.000000, 0.000000}
\pgfsetstrokecolor{dialinecolor}
\node[anchor=west] at (48.50\du,21.20\du){$1$};
\end{tikzpicture}

%% file: figures/TNsim.tex
% Graphic for TeX using PGF
% Title: /home/kim/Dropbox/Work/Papers/Journal/19.TNMoesp/figures/TNsim.dia
% Creator: Dia v0.97.2
% CreationDate: Tue Sep 12 08:31:42 2017
% For: kim
% \usepackage{tikz}
% The following commands are not supported in PSTricks at present
% We define them conditionally, so when they are implemented,
% this pgf file will use them.
\ifx\du\undefined
  \newlength{\du}
\fi
\setlength{\du}{4\unitlength}
\begin{tikzpicture}
\pgftransformxscale{1.000000}
\pgftransformyscale{-1.000000}
\definecolor{dialinecolor}{rgb}{0.000000, 0.000000, 0.000000}
\pgfsetstrokecolor{dialinecolor}
\definecolor{dialinecolor}{rgb}{1.000000, 1.000000, 1.000000}
\pgfsetfillcolor{dialinecolor}
\definecolor{dialinecolor}{rgb}{1.000000, 1.000000, 1.000000}
\pgfsetfillcolor{dialinecolor}
\pgfpathellipse{\pgfpoint{13.090120\du}{11.015600\du}}{\pgfpoint{3.550000\du}{0\du}}{\pgfpoint{0\du}{3.450000\du}}
\pgfusepath{fill}
\pgfsetlinewidth{0.100000\du}
\pgfsetdash{}{0pt}
\pgfsetdash{}{0pt}
\definecolor{dialinecolor}{rgb}{0.000000, 0.000000, 0.000000}
\pgfsetstrokecolor{dialinecolor}
\pgfpathellipse{\pgfpoint{13.090120\du}{11.015600\du}}{\pgfpoint{3.550000\du}{0\du}}{\pgfpoint{0\du}{3.450000\du}}
\pgfusepath{stroke}
\pgfsetlinewidth{0.100000\du}
\pgfsetdash{}{0pt}
\pgfsetdash{}{0pt}
\pgfsetbuttcap
{
\definecolor{dialinecolor}{rgb}{0.000000, 0.000000, 0.000000}
\pgfsetfillcolor{dialinecolor}
% was here!!!
\definecolor{dialinecolor}{rgb}{0.000000, 0.000000, 0.000000}
\pgfsetstrokecolor{dialinecolor}
\draw (16.640100\du,11.015600\du)--(21.013300\du,11.015600\du);
}
\pgfsetlinewidth{0.100000\du}
\pgfsetdash{}{0pt}
\pgfsetdash{}{0pt}
\pgfsetbuttcap
{
\definecolor{dialinecolor}{rgb}{0.000000, 0.000000, 0.000000}
\pgfsetfillcolor{dialinecolor}
% was here!!!
\definecolor{dialinecolor}{rgb}{0.000000, 0.000000, 0.000000}
\pgfsetstrokecolor{dialinecolor}
\draw (5.395120\du,11.005600\du)--(9.540120\du,11.015600\du);
}
\pgfsetlinewidth{0.100000\du}
\pgfsetdash{}{0pt}
\pgfsetdash{}{0pt}
\pgfsetbuttcap
{
\definecolor{dialinecolor}{rgb}{0.000000, 0.000000, 0.000000}
\pgfsetfillcolor{dialinecolor}
% was here!!!
\definecolor{dialinecolor}{rgb}{0.000000, 0.000000, 0.000000}
\pgfsetstrokecolor{dialinecolor}
\draw (13.090100\du,14.465600\du)--(13.110300\du,19.045200\du);
}
\pgfsetlinewidth{0.100000\du}
\pgfsetdash{}{0pt}
\pgfsetdash{}{0pt}
\pgfsetbuttcap
{
\definecolor{dialinecolor}{rgb}{0.000000, 0.000000, 0.000000}
\pgfsetfillcolor{dialinecolor}
% was here!!!
\definecolor{dialinecolor}{rgb}{0.000000, 0.000000, 0.000000}
\pgfsetstrokecolor{dialinecolor}
\draw (13.095100\du,3.035000\du)--(13.090100\du,7.565600\du);
}
% setfont left to latex
\definecolor{dialinecolor}{rgb}{0.000000, 0.000000, 0.000000}
\pgfsetstrokecolor{dialinecolor}
\node[anchor=west] at (12.60\du,5.50\du){$m$};
\definecolor{dialinecolor}{rgb}{1.000000, 1.000000, 1.000000}
\pgfsetfillcolor{dialinecolor}
\pgfpathellipse{\pgfpoint{24.563300\du}{11.015600\du}}{\pgfpoint{3.550000\du}{0\du}}{\pgfpoint{0\du}{3.450000\du}}
\pgfusepath{fill}
\pgfsetlinewidth{0.100000\du}
\pgfsetdash{}{0pt}
\pgfsetdash{}{0pt}
\definecolor{dialinecolor}{rgb}{0.000000, 0.000000, 0.000000}
\pgfsetstrokecolor{dialinecolor}
\pgfpathellipse{\pgfpoint{24.563300\du}{11.015600\du}}{\pgfpoint{3.550000\du}{0\du}}{\pgfpoint{0\du}{3.450000\du}}
\pgfusepath{stroke}
\pgfsetlinewidth{0.100000\du}
\pgfsetdash{}{0pt}
\pgfsetdash{}{0pt}
\pgfsetbuttcap
{
\definecolor{dialinecolor}{rgb}{0.000000, 0.000000, 0.000000}
\pgfsetfillcolor{dialinecolor}
% was here!!!
\definecolor{dialinecolor}{rgb}{0.000000, 0.000000, 0.000000}
\pgfsetstrokecolor{dialinecolor}
\draw (24.563300\du,14.465600\du)--(24.534000\du,18.985700\du);
}
\pgfsetlinewidth{0.100000\du}
\pgfsetdash{}{0pt}
\pgfsetdash{}{0pt}
\pgfsetbuttcap
{
\definecolor{dialinecolor}{rgb}{0.000000, 0.000000, 0.000000}
\pgfsetfillcolor{dialinecolor}
% was here!!!
\definecolor{dialinecolor}{rgb}{0.000000, 0.000000, 0.000000}
\pgfsetstrokecolor{dialinecolor}
\draw (24.540100\du,3.045000\du)--(24.563300\du,7.565600\du);
}
\pgfsetlinewidth{0.100000\du}
\pgfsetdash{}{0pt}
\pgfsetdash{}{0pt}
\pgfsetbuttcap
{
\definecolor{dialinecolor}{rgb}{0.000000, 0.000000, 0.000000}
\pgfsetfillcolor{dialinecolor}
% was here!!!
\definecolor{dialinecolor}{rgb}{0.000000, 0.000000, 0.000000}
\pgfsetstrokecolor{dialinecolor}
\draw (38.137400\du,11.169800\du)--(42.510600\du,11.169800\du);
}
\pgfsetlinewidth{0.100000\du}
\pgfsetdash{}{0pt}
\pgfsetdash{}{0pt}
\pgfsetbuttcap
{
\definecolor{dialinecolor}{rgb}{0.000000, 0.000000, 0.000000}
\pgfsetfillcolor{dialinecolor}
% was here!!!
\definecolor{dialinecolor}{rgb}{0.000000, 0.000000, 0.000000}
\pgfsetstrokecolor{dialinecolor}
\draw (28.113300\du,11.015600\du)--(32.413900\du,11.015600\du);
}
\definecolor{dialinecolor}{rgb}{1.000000, 1.000000, 1.000000}
\pgfsetfillcolor{dialinecolor}
\pgfpathellipse{\pgfpoint{46.060600\du}{11.169800\du}}{\pgfpoint{3.550000\du}{0\du}}{\pgfpoint{0\du}{3.450000\du}}
\pgfusepath{fill}
\pgfsetlinewidth{0.100000\du}
\pgfsetdash{}{0pt}
\pgfsetdash{}{0pt}
\definecolor{dialinecolor}{rgb}{0.000000, 0.000000, 0.000000}
\pgfsetstrokecolor{dialinecolor}
\pgfpathellipse{\pgfpoint{46.060600\du}{11.169800\du}}{\pgfpoint{3.550000\du}{0\du}}{\pgfpoint{0\du}{3.450000\du}}
\pgfusepath{stroke}
\pgfsetlinewidth{0.100000\du}
\pgfsetdash{}{0pt}
\pgfsetdash{}{0pt}
\pgfsetbuttcap
{
\definecolor{dialinecolor}{rgb}{0.000000, 0.000000, 0.000000}
\pgfsetfillcolor{dialinecolor}
% was here!!!
\definecolor{dialinecolor}{rgb}{0.000000, 0.000000, 0.000000}
\pgfsetstrokecolor{dialinecolor}
\draw (46.060600\du,14.619800\du)--(46.077100\du,19.134300\du);
}
\pgfsetlinewidth{0.100000\du}
\pgfsetdash{}{0pt}
\pgfsetdash{}{0pt}
\pgfsetbuttcap
{
\definecolor{dialinecolor}{rgb}{0.000000, 0.000000, 0.000000}
\pgfsetfillcolor{dialinecolor}
% was here!!!
\definecolor{dialinecolor}{rgb}{0.000000, 0.000000, 0.000000}
\pgfsetstrokecolor{dialinecolor}
\draw (45.990941\du,3.354414\du)--(46.060600\du,7.719800\du);
}
\pgfsetlinewidth{0.100000\du}
\pgfsetdash{}{0pt}
\pgfsetdash{}{0pt}
\pgfsetbuttcap
{
\definecolor{dialinecolor}{rgb}{0.000000, 0.000000, 0.000000}
\pgfsetfillcolor{dialinecolor}
% was here!!!
\definecolor{dialinecolor}{rgb}{0.000000, 0.000000, 0.000000}
\pgfsetstrokecolor{dialinecolor}
\draw (49.610600\du,11.169800\du)--(54.026800\du,11.169800\du);
}
% setfont left to latex
\definecolor{dialinecolor}{rgb}{0.000000, 0.000000, 0.000000}
\pgfsetstrokecolor{dialinecolor}
\node[anchor=west] at (10.30\du,10.90\du){$\ten{T}^{(1)}$};
% setfont left to latex
\definecolor{dialinecolor}{rgb}{0.000000, 0.000000, 0.000000}
\pgfsetstrokecolor{dialinecolor}
\node[anchor=west] at (21.80\du,10.90\du){$\ten{T}^{(2)}$};
% setfont left to latex
\definecolor{dialinecolor}{rgb}{0.000000, 0.000000, 0.000000}
\pgfsetstrokecolor{dialinecolor}
\node[anchor=west] at (24.10\du,5.50\du){$m$};
% setfont left to latex
\definecolor{dialinecolor}{rgb}{0.000000, 0.000000, 0.000000}
\pgfsetstrokecolor{dialinecolor}
\node[anchor=west] at (45.50\du,5.50\du){$m$};
% setfont left to latex
\definecolor{dialinecolor}{rgb}{0.000000, 0.000000, 0.000000}
\pgfsetstrokecolor{dialinecolor}
\node[anchor=west] at (12.60\du,17.50\du){$1$};
% setfont left to latex
\definecolor{dialinecolor}{rgb}{0.000000, 0.000000, 0.000000}
\pgfsetstrokecolor{dialinecolor}
\node[anchor=west] at (24.10\du,17.50\du){$1$};
% setfont left to latex
\definecolor{dialinecolor}{rgb}{0.000000, 0.000000, 0.000000}
\pgfsetstrokecolor{dialinecolor}
\node[anchor=west] at (43.171900\du,10.80\du){$\ten{T}^{(d)}$};
% setfont left to latex
\definecolor{dialinecolor}{rgb}{0.000000, 0.000000, 0.000000}
\pgfsetstrokecolor{dialinecolor}
\node[anchor=west] at (6.25\du,9.80\du){$1$};
% setfont left to latex
\definecolor{dialinecolor}{rgb}{0.000000, 0.000000, 0.000000}
\pgfsetstrokecolor{dialinecolor}
\node[anchor=west] at (50.25\du,9.80\du){$1$};
% setfont left to latex
\definecolor{dialinecolor}{rgb}{0.000000, 0.000000, 0.000000}
\pgfsetstrokecolor{dialinecolor}
\node[anchor=west] at (17.25\du,9.80\du){$r_2$};
% setfont left to latex
\definecolor{dialinecolor}{rgb}{0.000000, 0.000000, 0.000000}
\pgfsetstrokecolor{dialinecolor}
\node[anchor=west] at (28.50\du,9.80\du){$r_3$};
% setfont left to latex
\definecolor{dialinecolor}{rgb}{0.000000, 0.000000, 0.000000}
\pgfsetstrokecolor{dialinecolor}
\node[anchor=west] at (38.70\du,9.80\du){$r_d$};
\pgfsetlinewidth{0.100000\du}
\pgfsetdash{}{0pt}
\pgfsetdash{}{0pt}
\pgfsetbuttcap
\pgfsetmiterjoin
\pgfsetlinewidth{0.100000\du}
\pgfsetbuttcap
\pgfsetmiterjoin
\pgfsetdash{}{0pt}
\definecolor{dialinecolor}{rgb}{0.000000, 0.000000, 0.000000}
\pgfsetfillcolor{dialinecolor}
\pgfpathellipse{\pgfpoint{33.951900\du}{11.075000\du}}{\pgfpoint{0.425000\du}{0\du}}{\pgfpoint{0\du}{0.425000\du}}
\pgfusepath{fill}
\definecolor{dialinecolor}{rgb}{0.000000, 0.000000, 0.000000}
\pgfsetstrokecolor{dialinecolor}
\pgfpathellipse{\pgfpoint{33.951900\du}{11.075000\du}}{\pgfpoint{0.425000\du}{0\du}}{\pgfpoint{0\du}{0.425000\du}}
\pgfusepath{stroke}
\pgfsetbuttcap
\pgfsetmiterjoin
\pgfsetdash{}{0pt}
\definecolor{dialinecolor}{rgb}{0.000000, 0.000000, 0.000000}
\pgfsetstrokecolor{dialinecolor}
\pgfpathellipse{\pgfpoint{33.951900\du}{11.075000\du}}{\pgfpoint{0.425000\du}{0\du}}{\pgfpoint{0\du}{0.425000\du}}
\pgfusepath{stroke}
\pgfsetlinewidth{0.100000\du}
\pgfsetdash{}{0pt}
\pgfsetdash{}{0pt}
\pgfsetbuttcap
\pgfsetmiterjoin
\pgfsetlinewidth{0.100000\du}
\pgfsetbuttcap
\pgfsetmiterjoin
\pgfsetdash{}{0pt}
\definecolor{dialinecolor}{rgb}{0.000000, 0.000000, 0.000000}
\pgfsetfillcolor{dialinecolor}
\pgfpathellipse{\pgfpoint{35.591900\du}{11.075000\du}}{\pgfpoint{0.425000\du}{0\du}}{\pgfpoint{0\du}{0.425000\du}}
\pgfusepath{fill}
\definecolor{dialinecolor}{rgb}{0.000000, 0.000000, 0.000000}
\pgfsetstrokecolor{dialinecolor}
\pgfpathellipse{\pgfpoint{35.591900\du}{11.075000\du}}{\pgfpoint{0.425000\du}{0\du}}{\pgfpoint{0\du}{0.425000\du}}
\pgfusepath{stroke}
\pgfsetbuttcap
\pgfsetmiterjoin
\pgfsetdash{}{0pt}
\definecolor{dialinecolor}{rgb}{0.000000, 0.000000, 0.000000}
\pgfsetstrokecolor{dialinecolor}
\pgfpathellipse{\pgfpoint{35.591900\du}{11.075000\du}}{\pgfpoint{0.425000\du}{0\du}}{\pgfpoint{0\du}{0.425000\du}}
\pgfusepath{stroke}
\pgfsetlinewidth{0.100000\du}
\pgfsetdash{}{0pt}
\pgfsetdash{}{0pt}
\pgfsetbuttcap
\pgfsetmiterjoin
\pgfsetlinewidth{0.100000\du}
\pgfsetbuttcap
\pgfsetmiterjoin
\pgfsetdash{}{0pt}
\definecolor{dialinecolor}{rgb}{0.000000, 0.000000, 0.000000}
\pgfsetfillcolor{dialinecolor}
\pgfpathellipse{\pgfpoint{37.146900\du}{11.075000\du}}{\pgfpoint{0.425000\du}{0\du}}{\pgfpoint{0\du}{0.425000\du}}
\pgfusepath{fill}
\definecolor{dialinecolor}{rgb}{0.000000, 0.000000, 0.000000}
\pgfsetstrokecolor{dialinecolor}
\pgfpathellipse{\pgfpoint{37.146900\du}{11.075000\du}}{\pgfpoint{0.425000\du}{0\du}}{\pgfpoint{0\du}{0.425000\du}}
\pgfusepath{stroke}
\pgfsetbuttcap
\pgfsetmiterjoin
\pgfsetdash{}{0pt}
\definecolor{dialinecolor}{rgb}{0.000000, 0.000000, 0.000000}
\pgfsetstrokecolor{dialinecolor}
\pgfpathellipse{\pgfpoint{37.146900\du}{11.075000\du}}{\pgfpoint{0.425000\du}{0\du}}{\pgfpoint{0\du}{0.425000\du}}
\pgfusepath{stroke}
\definecolor{dialinecolor}{rgb}{1.000000, 1.000000, 1.000000}
\pgfsetfillcolor{dialinecolor}
\pgfpathellipse{\pgfpoint{13.095100\du}{-0.415000\du}}{\pgfpoint{3.550000\du}{0\du}}{\pgfpoint{0\du}{3.450000\du}}
\pgfusepath{fill}
\pgfsetlinewidth{0.100000\du}
\pgfsetdash{}{0pt}
\pgfsetdash{}{0pt}
\definecolor{dialinecolor}{rgb}{0.000000, 0.000000, 0.000000}
\pgfsetstrokecolor{dialinecolor}
\pgfpathellipse{\pgfpoint{13.095100\du}{-0.415000\du}}{\pgfpoint{3.550000\du}{0\du}}{\pgfpoint{0\du}{3.450000\du}}
\pgfusepath{stroke}
\definecolor{dialinecolor}{rgb}{1.000000, 1.000000, 1.000000}
\pgfsetfillcolor{dialinecolor}
\pgfpathellipse{\pgfpoint{24.540100\du}{-0.405000\du}}{\pgfpoint{3.550000\du}{0\du}}{\pgfpoint{0\du}{3.450000\du}}
\pgfusepath{fill}
\pgfsetlinewidth{0.100000\du}
\pgfsetdash{}{0pt}
\pgfsetdash{}{0pt}
\definecolor{dialinecolor}{rgb}{0.000000, 0.000000, 0.000000}
\pgfsetstrokecolor{dialinecolor}
\pgfpathellipse{\pgfpoint{24.540100\du}{-0.405000\du}}{\pgfpoint{3.550000\du}{0\du}}{\pgfpoint{0\du}{3.450000\du}}
\pgfusepath{stroke}
\definecolor{dialinecolor}{rgb}{1.000000, 1.000000, 1.000000}
\pgfsetfillcolor{dialinecolor}
\pgfpathellipse{\pgfpoint{45.935100\du}{-0.145000\du}}{\pgfpoint{3.550000\du}{0\du}}{\pgfpoint{0\du}{3.450000\du}}
\pgfusepath{fill}
\pgfsetlinewidth{0.100000\du}
\pgfsetdash{}{0pt}
\pgfsetdash{}{0pt}
\definecolor{dialinecolor}{rgb}{0.000000, 0.000000, 0.000000}
\pgfsetstrokecolor{dialinecolor}
\pgfpathellipse{\pgfpoint{45.935100\du}{-0.145000\du}}{\pgfpoint{3.550000\du}{0\du}}{\pgfpoint{0\du}{3.450000\du}}
\pgfusepath{stroke}
% setfont left to latex
\definecolor{dialinecolor}{rgb}{0.000000, 0.000000, 0.000000}
\pgfsetstrokecolor{dialinecolor}
\node[anchor=west] at (11.20\du,-0.10\du){$\mat{u}_t$};
% setfont left to latex
\definecolor{dialinecolor}{rgb}{0.000000, 0.000000, 0.000000}
\pgfsetstrokecolor{dialinecolor}
\node[anchor=west] at (22.50\du,-0.10\du){$\mat{u}_t$};
% setfont left to latex
\definecolor{dialinecolor}{rgb}{0.000000, 0.000000, 0.000000}
\pgfsetstrokecolor{dialinecolor}
\node[anchor=west] at (44.00\du,-0.10\du){$\mat{u}_t$};
% setfont left to latex
\definecolor{dialinecolor}{rgb}{0.000000, 0.000000, 0.000000}
\pgfsetstrokecolor{dialinecolor}
\node[anchor=west] at (45.50\du,17.50\du){$p+n$};
\end{tikzpicture}